\documentclass[AMA,Times1COL]{WileyNJDv5} 
\usepackage{graphicx}
\usepackage{subfig}
\usepackage{amssymb,amsmath,amsthm}
\usepackage{mathrsfs} 
\usepackage{float}
\usepackage[section]{placeins}
\usepackage{dutchcal} 
\usepackage{color, colortbl}
\usepackage{makecell} %for multiple rows heading table
\usepackage[shortlabels]{enumitem}

\articletype{}%

\received{Date Month Year}
\revised{Date Month Year}
\accepted{Date Month Year}
\journal{Journal}
\volume{00}
\copyyear{2023}
\startpage{1}

\newcommand{\prob}{\mathbb{P}}
\newcommand{\ee}{\mathbb{E}}
\newcommand{\ind}[1]{I_{#1}}
\newcommand{\indd}[1]{\ind{\{#1\}}}

\newcommand{\cT}{\mathcal{T}}

\newcommand\dZ{\mathbb{Z}}

\newcommand\Tmax{T_{max}}
\newcommand{\har}{\widehat{\text{C}}^{\text{har}}_n}
\newcommand{\uno}{\widehat{\text{C}}^{\text{uno}}_n}
\newcommand{\tdc}{\widehat{\text{C}}_n}
\newcommand{\tdu}{\widehat{\text{C}}^w_n}

\theoremstyle{definition}

\theoremstyle{remark}

\theoremstyle{plain}

\newcommand\T{\mathbb{T}}

\begin{document}

\title{A new discrimination measure for assessing predictive performance of non-linear survival models}

\author[1,2]{Alfensi Faruk}

\author[1]{Jan Palczewski}

\author[1]{Georgios Aivaliotis}

\authormark{FARUK \textsc{et al.}}
\titlemark{A new discrimination measure for assessing predictive performance of non-linear survival models}

\address[1]{\orgdiv{Department of Statistics}, \orgname{University of Leeds}, \orgaddress{\state{Leeds}, \country{UK}}}

\address[2]{\orgdiv{Department of Mathematics}, \orgname{Universitas Sriwijaya}, \orgaddress{\state{Indralaya}, \country{Indonesia}}}

\corres{ \email{G.Aivaliotis@leeds.ac.uk}}

\presentaddress{University of Leeds, Leeds, LS2 9JT, UK}

%\fundingInfo{Text}
%\JELinfo{ejlje}

\abstract[Abstract]{Non-linear survival models are flexible models in which the proportional hazard assumption is not required. This poses difficulties in their evaluation. We introduce a new discrimination measure, time-dependent Uno's C-index, to assess the discrimination performance of non-linear survival models. This is an unbiased version of Antolini's time-dependent concordance. We prove convergence of both measures employing Nolan and Pollard's results on U-statistics. We explore the relationship between these measures and, in particular, the bias of Antolini's concordance in the presence of censoring using simulated data. We demonstrate the value of time-dependent Uno's C-index for the evaluation of models trained on censored real data and for model tuning.}


\jnlcitation{\cname{
\author{A. Faruk}, 
\author{J. Palczewski}, and 
\author{G. Aivaliotis}} (\cyear{xxxx}), 
\ctitle{A new discrimination measure for assessing the predictive performance of non-linear survival models}, \cjournal{Statistics in Medicine}, \cvol{xxxx;00:x--x}.}

\maketitle

\section{Introduction}\label{sec1}

Survival modelling, a celebrated area of success in statistics since the appearance of the Cox proportional hazards (CPH) model \cite{cox1972regression}, has been of immense importance to a number of applications, including in finance and medicine. Advances in statistical models for survival analysis \cite{kleinbaum2010survival,yang2013cocktail} have allowed for more complex relationships between covariates and the risk. They remained, however, focused on some form of (generalised) linear relationship. The introduction of machine learning (ML) and artificial intelligence (AI) freed the users from the need to make  linearity assumptions and allowed for complicated and potentially highly non-linear relationships between covariates and the risk. This enhanced the potential for accurate predictions albeit at the expense of interpretability and additional challenges in measuring the prediction accuracy.

Assessing the accuracy of risk predictions is crucial as it enhances users' confidence in their models and decisions. The literature contains a variety of measures of prediction accuracy. These measures fall broadly into two categories: calibration and discrimination. Calibration assesses how closely the number of predicted events matches the number of observed events over the follow-up time. Hosmer-Lemeshow goodness-of-fit test\cite{hosmer1980goodness}, D-calibration \cite{haider2020effective,goldstein2020x}, and Brier score \cite{graf1999assessment} are among the common calibration measures in survival analysis. Discrimination, on the other hand, assesses a model's ability to distinguish between higher and lower-risk individuals. For a review of discrimination measures for survival models see \citep{rahman2017review}. Concordance indexes try to estimate the probability of ``agreement'' between the risk prediction for a random pair of ``at risk'' individuals/objects and the actual outcome. Harrell's concordance index (C-index), \cite{harrell1982evaluating} is the most established one, requiring univariate predicted scores as the primary outputs of the evaluated models. C-index is more suitable for assessing the predictive performance of survival models with proportional hazard (PH) assumption, such as the CPH model and some common parametric survival models. The PH assumption states that the risks of two individuals experiencing the event of interest are always proportional to each other, so that a single risk score is capable of discriminating between more and less risky individuals. 
In problems with longer follow-up periods or with more complex relationship between covariates and the survival probability, it is less likely to maintain proportionality of risks. In such a situation, PH survival models and standard discrimination measures, e.g. Harrell's C-index, are no longer suitable.

One of the challenges arising with the introduction of non-linear survival models (such as ML and AI models but also statistical models) is the violation of the PH assumption. This allows for the possibility of a reversal of the relative risk relationship between two individuals. Therefore, discrimination measures based on evaluating the survival curve at one specific time point or abstracting the survival curve into a univariate risk score in a more general fashion are not suitable. Indeed, in non-PH survival models, the full predicted survival probability curve needs to enter into the evaluation. Antolini's \cite{antolini2005time} ``time-dependent'' version of Harrell's C-index called time-dependent concordance, lends itself to address this challenge as it uses the risk predicted by the model at the time of the first event between each pair of individuals. The time-dependent concordance is equivalently established as the weighted average of time-dependent area under curve (AUC) \cite{heagerty2005survival, pepe2008evaluating}, which is the area under receiver operating characteristics (ROC) curve drawn by plotting false-positive rate and true-positive rate in x-axis and y-axis, respectively. This is a promising measure addressing the limitations of Harrell's C-index, however no proof of convergence of the time-dependent concordance estimator to the probability of concordance is provided in the original paper by Antolini et al. \cite{antolini2005time} and in the remaining literature. We address this issue in the present paper and provide a detailed proof of the convergence of the original estimator by Antolini et al. \cite{antolini2005time}. 

Another well-documented problem with evaluation of survival models is a potential bias of estimators due to the presence of censoring in data \cite{gerds2013estimating}. Uno's C-index \cite{uno2011c} was developed to overcome the bias of the original Harrell's C-index by applying an inverse probability of censoring weighted (IPCW) approach, where the censoring time is independent from the covariates and failure times and identically distributed between individuals. Uno et al. \cite{uno2011c} show that their estimator converges in probability to the probability of concordance, generalising results of Harrell to the censored data. Several IPCW-based C-indexes for standard Harrell's C-index, where the censoring survival probability is conditional on covariates, are discussed by Gerds et al. \citep{gerds2013estimating}.  

This paper focuses on discrimination measures for non-linear models and right-censored data for a single event of interest. A novel estimator of time-dependent concordance is proposed; we call it time-dependent Uno's C-index as it uses the idea of IPCW as well. We show that this estimator is unbiased under the same conditions for censoring distribution as in the original Uno's paper. 
Our analysis also identifies a gap in the original result of Uno et al. \cite{uno2011c}; an additional condition on the censoring distribution is required for its validity. This condition ensures that the survival function of the censoring distribution is uniformly separated from zero, see Remark \ref{rem:uno}.

The measures discussed in this paper only use as input the event times, the censoring distribution and the predicted risk. The predicted risk can be the output of any relevant model and the model itself is not important. Although many ML approaches have been proposed to deal with right-censored survival data (see Wang et al.\cite{wang2019machine} for a complete review), neural networks and random (survival) forests (Breiman et al. \citep{breiman2001random}, Ishwaran et al.\cite{ishwaran2008random}) are still predominantly applied by medical researchers \cite{huang2023application}. Some state-of-the-art neural networks for survival analysis, such as DeepSurv\cite{katzman2018deepsurv}, Cox-nnet\cite{ching2018cox}, Nnet-survival\cite{gensheimer2019scalable}, can be found in current literature. Nnet-survival and RSF, are chosen as the primary examples of non-linear survival models to be evaluated using our proposed discrimination measures. 

We compare the performance of the time-dependent Uno's C-index with the time-dependent concordance and C-index in a number of experiments. We show how a non-PH scenario does not affect the time-dependent measures while it renders the original Harrell's C-index unusable and potentially liable for manipulation. Through a simulation study, where we control the rate of censoring in the test data, we establish the advantage of the IPCW estimators over the non-weighted ones. We investigate the relationship between the time-dependent Uno and the time-dependent concordance. We demonstrate that time-dependent concordance can be either upward or downward biased depending on the relationship between the distribution of censoring times and the accuracy of predicted survival curves at various times. We show that the time-dependent Uno's C-index is significantly less affected by the censoring rates in all cases; some effect of the censoring on the accuracy of estimators is always maintained due to the reduction of the amount of data in regions with high censoring probability. This disappears in our estimator when the amount of data increases but persists in the time-dependent concordance. Our findings shed light on the relationship between the time-dependent Uno's C-index and the time-dependent concordance that we observe in the Heart Failure data\cite{ahmad2017survival}. 

We apply our concordance measure to evaluate the predictive performance of non-linear survival models trained on two real datasets: Heart Failure \cite{ahmad2017survival} and The Cancer Genome Atlas mutations \cite{kandoth2013mutational}.
We demonstrate the applicability of our measure to the evaluation of models and to hyper-parameter tuning.

This paper is structured as follows. Section~\ref{sec2} presents the modelling framework, existing evaluation measures for survival models and their properties and derives our time-dependent Uno's C-index. In Section~\ref{sec3} we perform a thorough analysis of evaluation measures on simulated datasets. Section~\ref{sec4} contains studies on real datasets. Section~\ref{sec5} discusses our results in the backdrop of the relevant literature. Proofs of Theorems are collected in the Appendix~\ref{app1}.

%-----------------------------------------------
\section{Discrimination measures}\label{sec2}

Let the event time of an individual $i$, $i=1,\cdots,n$, be denoted by $T_i$, either a random variable or its realisation in the data, depending on the context. Throughout this paper, a survival model is defined as a function 
\begin{equation*}
    S: (t,\dZ_i) \mapsto [0,1],
\end{equation*}
where $\dZ_i=[Z_{i1} \cdots Z_{ip}]'$ is $(p \times 1)$ vector of covariates of individual $i$. The predicted survival curve, $S(t;\dZ_i)$, is the predicted probability of $T_i > t$. It can be interpreted as the predicted probability of individual $i$ surviving from the beginning of the follow-up up to time $t \in [0,\Tmax]$.

Let $D_i$ be the censoring time of individual $i$, which is assumed to be independent between individuals, independent from the event time $T_i$ and the covariates with some unknown but fixed distribution with the tail distribution function $G(t) = \prob(D_i > t)$. The notation $D_i$ will either stand for a random variable or its realisation in the data. We write the observed time as $X_i = T_i \wedge D_i := \min(T_i, D_i)$ -- this is the time available from the data; it is either the event time if it happens before censoring or, otherwise, the censoring time. We will use subscripts $i$, $j$ to denote two randomly chosen individuals from the whole population or the sample data depending on the context. Note that any individual $i$ with $X_i=\Tmax$ is administratively censored.

Throughout the paper, we make the following assumptions about the population distribution:

\begin{enumerate}[start=1,label={(R\arabic*):}]
\item $(\dZ_i, T_i, D_i)$ are i.i.d. with an (unknown) distribution $F$,
	\item $(D_i)$ are independent from $(\dZ_i, T_i)$ with an (unknown) tail distribution function $G$.
\end{enumerate}

\begin{remark}\label{remark1}
    When dealing with discrete-time units, the realisation of $T_i, D_i$, and $X_i$ belongs to a set of possible event times, $\cT = \{1, \cdots, \Tmax\}$, where the elements of $\cT$ are referred to as `periods'. $\cT$ may be dictated by the data or obtained by discretising a follow-up $[0,\text{B}]$ into $\Tmax$ time intervals $[a_0,a_1),[a_1,a_2),\cdots,[a_{T_{max}},\infty)$, where $a_0=0$, and $\Tmax$ is the period in which all individuals are administratively censored. We often discretise data into a sufficiently small number of periods as the structure of many standard discrete-time survival models cannot handle a large number of periods. Throughout this work, unless it is stated that a discrete-time setting applies, we consider a continuous-time setting.
\end{remark}

 In the sections that follow, we assume that the, possibly non-linear, survival model $S$ is given and we assess its performance using a test dataset of size $n$, independent of the training set. The tail distribution of censoring times $G$ is approximated by an empirical tail function $\widehat G_n$ based on observed censoring times $\{D_i,\cdots,D_n\}$ in the test data.
 Throughout this paper, we use KM estimator to compute $\widehat G_n$. The subindex $n$ on the index notation indicates its sample estimator.

\subsection{Existing Discrimination Measures}\label{subsec:discrimination}

This section introduces several discrimination measures for assessing the discrimination performance of non-linear survival models.
We first adapt standard measures, such as Harrell's C-index and Uno's C-index, to a fixed time point $t$, making it suitable for discrete time problems. We then discuss the time-dependent concordance and propose its IPCW generalisation, namely time-dependent Uno.

%-----------------------------------------------
\subsubsection{Harrell's C-Index}\label{sec:harrell}

Harrell's C-index \cite{harrell1982evaluating} was developed for assessing the discrimination capability of a model. It is defined by the ratio of the number of all concordant pairs to the number of all ``usable'' pairs:
\begin{equation}\label{eqn:cont_harr}
    \har=\frac{\sum_{i\neq j}^n\indd{T_i < D_i}\indd{\widehat{\mu}_i  > \widehat{\mu}_j }\indd{T_i<X_j}}{\sum_{i\neq j}^n \indd{T_i < D_i}\indd{T_i<X_j}},
\end{equation}
where $\widehat{\mu}_i$ and $\widehat{\mu}_j$ are the predicted risk scores of individuals $i$ and $j$. An ordered pair of two different individuals $(i \ne j)$ with the respective observed times $X_i$ and $X_j$  is usable if both individuals are uncensored, or the first is uncensored and the censored time of the other is larger than the uncensored time. 

Applying the idea in \eqref{eqn:cont_harr} to potentially non-PH survival models, although not recommended as we will explain later, requires specification of the risk score; we take the predicted survival curve at a fixed time $t$. Note that using survival curve instead of risk score in C-index is common for the evaluation of survival models\cite{antolini2005time,gensheimer2019scalable}. 
By applying the predicted survival curve as the evaluated quantity, the adaptation of \eqref{eqn:cont_harr} to a fixed time $t$ is defined as follows 
\begin{equation}\label{eqn:trunc_harr}
    \har(t)=\frac{\sum_{i\neq j}^n\indd{T_i < D_i}\indd{S(t; \dZ_i)  < S(t; \dZ_j) }\indd{T_i<X_j}}{\sum_{i\neq j}^n \indd{T_i < D_i}\indd{T_i<X_j}},
\end{equation}
where we define that the usable pair is concordant if $S(t; \dZ_i)$ is less than $S(t; \dZ_j)$ when $T_i < X_j$. 

According to Uno et al.\cite{uno2011c}, as $n \to \infty$, $\har$ converges in probability to 
\begin{equation}\label{eqn:prob_har}
    \prob\left[\mu_i > \mu_j \big|T_i < T_j,T_i < D_i \wedge D_j\right].
\end{equation}
 Similarly, $\har(t)$ will converge in probability as $n \to \infty$ to
\begin{equation}\label{eqn:prob_har_t}
    \prob\left[S(t; \dZ_i) < S(t; \dZ_j)\big|T_i < T_j,T_i < D_i \wedge D_j \right].
\end{equation}

%============================
\subsubsection{Uno's C-Index}\label{sec:uno}
Harrell's C-index computes the concordance for pairs for which $T_i < D_i\wedge D_j$, which includes elements of the censoring distribution. Uno et al.\cite{uno2011c} proposed the IPCW-based C-index:
\begin{equation}\label{eqn:cont_uno}
    \uno=\frac{\sum_{i\neq j}^n\indd{T_i < D_i}\indd{\widehat{\mu}_i  > \widehat{\mu}_j }\indd{T_i<X_j}\widehat{G}_n^{-2}(T_i)}{\sum_{i\neq j}^n \indd{T_i < D_i}\indd{T_i<X_j}\widehat{G}_n^{-2}(T_i)}.
\end{equation}
Uno's C-index evaluated at a specific time $t \in [0,\Tmax)$ is defined as follows:
\begin{equation}\label{eqn:uno_t}
    \uno(t)=\frac{\sum_{i\neq j}^n \indd{T_i< D_i}  \indd{S(t; \dZ_i) <S(t; \dZ_j) }\indd{T_i<X_j, T_i < \Tmax}\widehat{G}_n^{-2}(T_i)}{\sum_{i\neq j}^n \indd{T_i < D_i}\indd{T_i<X_j,T_i < \Tmax}\widehat{G}_n^{-2}(T_i)},
\end{equation}
where $\widehat{G}_n(T_i)$ is the predicted survival curve computed by KM estimator for censoring times based on $n$ samples from the test data. $\Tmax$ is a pre-specified cut-off time such that $G(\Tmax)>0$. Both $\uno$ and $\uno(t)$ are unbiased converging in probability as $n \to \infty$ to the population probabilities\citep{uno2011c} respectively
\begin{equation}\label{eqn:prob_uno}
    \prob\left[\mu_i > \mu_j \big|T_i < T_j,T_i < \Tmax\right],
\end{equation}
and
\begin{equation}\label{eqn:prob_uno_t}
	\prob\left[S(t; \dZ_i) < S(t; \dZ_j)\big|T_i < T_j,T_i < \Tmax\right].
\end{equation}

\begin{remark}\label{rem:uno}
As we discuss later in Remark \ref{rem:epsilon} and can be seen in the proof of Theorem \ref{theo:conv_tdu}, the almost sure convergence of the above estimators require that there is an priori known constant $\epsilon > 0$ such that $\widehat G_n \ge \epsilon$. As far as we could see in the original paper of Uno et al.\cite{uno2011c}, their arguments also require this assumption for their validity. This is related with the use of the uniform convergence results for U-statistics.

\end{remark}

%============================
\subsubsection{Time-Dependent Concordance}\label{subsect:tdc}

Antolini et al.\cite{antolini2005time} used the weighted time-dependent AUC to derive a discrimination measure for models that use time-dependent covariates, therefore suitable for survival models without the PH assumption. Adjusting to our setting (truncating at $T_i < \Tmax$), their measure estimates
\begin{equation}\label{eqn:ctd_antolini}
    \begin{aligned}
    \text{C}
    &=
    \prob\big[S(T_i;\dZ_i)<S(T_i;\dZ_j)\big| T_i<T_j, T_i < \Tmax\big]\\
    &= 
    \frac{\prob\big[S(T_i; \dZ_i) < S(T_i;\dZ_j), T_i<T_j, T_i < \Tmax\big]}{\prob\big[T_i<T_j, T_i < \Tmax\big]}\\
    &=
    \frac{\ee\big[ \indd{S(T_i;\dZ_i)<S(T_i;\dZ_j)}\indd{T_i < T_j, T_i < \Tmax}\big]}{\ee\big[\indd{T_i < T_j, T_i < \Tmax}\big]}.
    \end{aligned}
\end{equation}

The estimator of \eqref{eqn:ctd_antolini} can be derived by a direct approximation of the numerator and denominator of the last equality of \eqref{eqn:ctd_antolini} in the presence of censoring as follows
\begin{equation}\label{eqn:est_tdc}
	\tdc=\frac{\sum_{i\neq j}^n \indd{T_i < D_i}\indd{S(T_i;\dZ_i)<S(T_i;\dZ_j)}\indd{T_i<X_j, T_i < \Tmax}}{\sum_{i\neq j}^n \indd{T_i < D_i}\indd{T_i<X_j, T_i < \Tmax}},
\end{equation}
where we refer to $\tdc$ as time-dependent concordance. However, Antolini et al.\cite{antolini2005time} do not provide proof of convergence of time-dependent concordance. We state the convergence of $\tdc$ in Theorem~\ref{lem1}, where the corresponding proof can be seen in Appendix~\ref{app1.1a}. In Section~\ref{subsect:tdu}, we will amend $\tdc$ so that the correct convergence is obtained.

%----------------------------
\begin{theorem}[Convergence of Time-Dependent Concordance]\label{lem1}
Assume (R1-R2). Then,
	\begin{equation*}
			\tdc \overset{a.s.}{\to} \prob\big[S(T_i;\dZ_i)<S(T_i;\dZ_j)\big| T_i<T_j, T_i < \Tmax,  T_i<D_i\wedge D_j\big].
	\end{equation*}
\end{theorem}

Time-dependent concordance is designed to evaluate the model performance based on predicted survival curves at event times and summarises the overall model performance. It provides a more flexible measure which is suitable for more general survival models in which the PH assumption is not imposed. However, time-dependent concordance is strongly related to Harrell's C-index since they are developed under the same principles except for the evaluated times and the PH assumption. As a consequence, time-dependent concordance suffers from similar drawbacks as Harrell's C-index, in particular its bias increases as the censoring rate increases \citep{gerds2013estimating}. Uno's C-Index administered an IPCW approach to alleviate the dependence of Harrell's C-index on the censoring distribution. Note that the formula, the convergence, and the convergence proof of time-dependent concordance for discrete-time units are the same as the continuous-time units. We only need to consider that the realisation of $T_i, D_i$, and $X_i$ belongs to $\cT$.

%----------------------------
\subsection{Time-Dependent Uno's C-index}\label{subsect:tdu}
To derive the formula of time-dependent Uno's C-index, we apply a different approach to the original paper of time-dependent concordance\cite{antolini2005time}. We first recall the last equality of \eqref{eqn:ctd_antolini} as follows 
\begin{equation}\label{eqn:cw}
        \text{C} = \frac{\ee\big[ \indd{S(T_i;\dZ_i)<S(T_i;\dZ_j)}\indd{T_i < T_j, T_i < \Tmax}\big]}{\ee\big[\indd{T_i < T_j, T_i < \Tmax}\big]}.
\end{equation}
Note that the expected value of $\big(\indd{T_i < D_i}\big/G(T_i) \big|T_i,\dZ_i,T_j,\dZ_j,D_j\big)$, where   $D_i$  is independent of $T_i,\dZ_i,T_j,\dZ_j$, and $D_j$, is given by
\begin{equation}\label{exp_indentity}
    \ee\bigg[\frac{\indd{T_i < D_i}}{G(T_i)} \big|T_i,\dZ_i,T_j,\dZ_j,D_j\bigg]
    =
    \frac{1}{G(T_i)}\ee\big[\indd{T_i< D_i}\big|T_i,\dZ_i,T_j,\dZ_j,D_j\big]
    =
    \frac{1}{G(T_i)}\ee\big[\indd{T_i< D_i}\big|T_i\big]
    =
    \frac{1}{G(T_i)}G(T_i)
    =
    1.
\end{equation}

Using \eqref{exp_indentity}, the numerator of the right-hand side of \eqref{eqn:cw} can be rewritten as follows:
\begin{equation}\label{eqn:exp_num_tdu}
    \ee\bigg[\indd{S(T_i;\dZ_i) < S(T_i;\dZ_j)}\indd{T_i < T_j,T_i < \Tmax} \frac{\indd{T_i < D_i}}{G(T_i)} \frac{\indd{T_i < D_j}}{G(T_i)}\bigg] = \ee\bigg[\indd{S(T_i;\dZ_i) < S(T_i;\dZ_j)}\indd{T_i < X_j,T_i < \Tmax} \indd{T_i < D_i}G^{-2}(T_i)\bigg],    
\end{equation}
where $G^{-2}(T_i)$ is the ``penalty" term at $T_i$ due to ignoring some unusable pairs. 

By the same arguments, the denominator of \eqref{eqn:cw} can also be written as
\begin{equation}\label{eqn:exp_den_tdu}
    \ee\bigg[\indd{T_i < T_j,T_i < \Tmax} \frac{\indd{T_i < D_i}}{G(T_i)} \frac{\indd{T_i < D_j}}{G(T_i)}\bigg] = \ee\bigg[\indd{T_i < X_j,T_i < \Tmax} \indd{T_i < D_i}G^{-2}(T_i)\bigg].
\end{equation}
From \eqref{eqn:exp_num_tdu} and \eqref{eqn:exp_den_tdu}, we obtain
\begin{equation}\label{eqn:C_rewrite}
    \text{C}  = \frac{\ee\big[\indd{S(T_i;\dZ_i) < S(T_i;\dZ_j)}\indd{T_i < X_j,T_i < \Tmax} \indd{T_i < D_i}G^{-2}(T_i)\big]}{\ee\big[\indd{T_i < X_j,T_i < \Tmax} \indd{T_i < D_i}G^{-2}(T_i)\big]}.
\end{equation}
Therefore, following \eqref{eqn:C_rewrite}, we finally propose a new estimator of $\text{C}$ called time-dependent Uno's C-index as follows:
\begin{equation}\label{eqn:tdu}
	\tdu=\frac{\sum_{i\neq j}^n \indd{T_i < D_i} \indd{S(T_i;\dZ_i) < S(T_i;\dZ_j)} \indd{T_i < X_j,T_i < \Tmax}\widehat{G}_n^{-2}(T_{i})} {\sum_{i\neq j}^n \indd{T_i < D_i}\indd{T_i < X_j,T_i < \Tmax}\widehat{G}_n^{-2}(T_{i})}.
\end{equation}
Note that $\tdu$ can be also seen as the weighted version of the time-dependent concordance, $\tdc$. The convergence of $\tdu$ is stated in Theorem~\ref{theo:conv_tdu} and the proof is given in Appendix~\ref{app1.1b}.

%----------------------------------------
\begin{theorem}[Convergence of Time-dependent Uno's C-index]\label{theo:conv_tdu}
Assume conditions (R1-R2) and that there exists $\epsilon > 0$ such that $\widehat G_n, G \in \mathcal{G}_\epsilon$, where
\begin{equation}\label{eqn:G_epsilon}
\mathcal{G}_{\epsilon} := \big\{ g:[0, T_{max}) \to [\epsilon, 1]\ \big|\  \text{$g(0) = 1$ and $g$ is non-increasing and right-continuous} \big\}.
\end{equation}

If $\sup_{t \in [0, T_{max})} |\widehat G_n(t) - G(t)| \to 0$ a.s. as $n \to \infty$, then
	\begin{equation*}
		 \tdu \overset{a.s.}{\to}  \text{C}, \quad n \to \infty.
	\end{equation*}
\end{theorem}

\begin{remark}
The assumption $\sup_{t \in [0, T_{max})} |\widehat G_n(t) - G(t)| \to 0$ a.s. as $n \to \infty$ is satisfied when $\widehat G_n$ is the empirical cumulative distribution function, see Glivenko-Cantelli Theorem\cite[Thm.~20.6]{Billingsley}. It is also implied by the following weaker condition: $\widehat G_n$ are right-continuous and non-increasing random functions such that $\widehat G_n(t) \to G(t)$ a.s. for every $t \in [0, T_{max})$. Indeed, taking a countable set $D \subset [0, T_{max})$ which is dense in $[0, T_{max})$, we have $\sup_{t \in D}|\widehat G_n(t) - G(t)| \to 0$ a.s. as $n \to \infty$. This convergence is extended to the whole interval $[0, T_{max})$ by the right-continuity.
\end{remark}

\begin{remark}
    The formula of time-dependent Uno's C-index for discrete-time units is the same as for continuous-time units given in \eqref{eqn:tdu}. However, we need to redefine the family of functions $\mathcal{G}_\epsilon$ in \eqref{eqn:G_epsilon}
    as follows
    \begin{equation}\label{eqn:G_discrete}
        \widetilde{\mathcal{G}}_{\epsilon}=\{g:\{1,\cdots,\Tmax-1\} \to [\epsilon,1]\ \big|\  \text{$g(0) = 1$ and $g$ is non-increasing} \big\}.
    \end{equation}
    and the assumption about the convergence of $\widehat G_n$ to $G$ reads: $\widehat G_n(t) \to G(t)$ a.s. for every $t \in \{1,2, \cdots, T_{max}-1\}$ (which is clearly satisfied when $\widehat G_n$ is the empirical cumulative distribution function). Comments about the proof are given in Appendix~\ref{app1.1c}.
\end{remark}

\begin{remark}\label{rem:epsilon}
The assumption that $\widehat G_n, G \in \mathcal{G}_\epsilon$ for $\epsilon > 0$ is key in the proof of the theorem. Without it, we would not be able to apply the uniform law of large numbers for U-statistics. A careful inspection of the proof would reveal that one could replace it with an assumption that $G(T_{max}-) > 0$. In the estimator, one would then use functions $\hat g_n (t) = \max(\widehat G_n(t), a_n)$ for an appropriately chosen sequence $a_n \to 0$. Incorporation of this generalisation to the proof would make it unnecessarily long and would not improve significantly its applicability to real-life problems, so we decided to omit it. Tuning of $\epsilon$ as discussed below seems to us a much more practical approach.
\end{remark}

\begin{remark}
In practice, we need to choose the lower bound $\epsilon$ in $\mathcal{G}_\epsilon$ \emph{a priori} as the true censoring distribution is unknown. If we choose $\epsilon$ too large, it will tune down the effect of $G$ while $\epsilon$ too small will increase the variance of the estimator by allowing an excess effect of underestimation of $\widehat G_n$ on the value of the statistic. Small values of $\widehat G_n$ blow up the contributions of the respective individuals to the measure.
\end{remark}

\begin{remark}
$G(T_i)$ accounts for the skewed distribution of event times $T_i$ due to censoring. Smaller values of $G(T_i)$ indicate smaller probability of seeing individuals with this event time compared to the uncensored case. Hence, their contributions to the statistic must be inflated compared to individuals with large $G(T_i)$.
\end{remark}

%----------------------------------------
\subsection{Relationship between Time-dependent Uno's C-index and Time-dependent Concordance}\label{sec2.3}

Exploring the relationship between the IPCW (weighted) and the corresponding non-IPCW (unweighted) discrimination measures can shed light on how to quantify the difference between those two measures. 

We rewrite $\text{C}$ as follows
\begin{equation}\label{eqn:tdu_tdc_rlt}
	\begin{aligned}
		\text{C}
		&=
		\prob\big[S(T_i;\dZ_i) < S(T_i;\dZ_j) \big|T_{i} < T_{j}, T_{i}  < \Tmax\big] \\
		&=
		\prob\big[S(T_i;\dZ_i) < S(T_i;\dZ_j) \big|T_{i} < T_{j}, T_i < \Tmax, T_i<D_i \wedge D_j \big] \prob\big[T_i < D_i \wedge D_j \big| T_{i} < T_{j}, T_i < \Tmax\big]\\
		&\quad+
		\prob\big[S(T_i;\dZ_i) < S(T_i;\dZ_j) \big|T_{i} < T_{j}, T_i < \Tmax, T_i\geq D_i \wedge D_j\big]\big(1-\prob\big[ T_i < D_i \wedge D_j \big| T_{i} < T_{j} , T_i < \Tmax \big]\big).
	\end{aligned}
\end{equation}
Recall also that
\begin{equation*}
	\text{C}^{td}=\prob\big[S(T_i;\dZ_i) < S(T_i;\dZ_j) \big|T_{i} < T_{j}, T_i < \Tmax, T_i<D_i \wedge D_j \big].
\end{equation*}
Denote
\[
\theta = \prob\big[T_i < D_i \wedge D_j \big| T_{i} < T_{j}, T_i < \Tmax\big] = \ee \big[ G^2(T_i) \big| T_{i} < T_{j}, T_i < \Tmax\big].
\]
Then \eqref{eqn:tdu_tdc_rlt} can be rewritten as
\begin{equation}\label{eqn:tdu_tdc_rlt2}
	\text{C} = \theta \text{C}^{td} + (1 - \theta) \tilde{\text{C}},
\end{equation}
where 
\begin{equation}\label{eqn:C_tilde}
	\tilde{\text{C}} = \prob\big[S(T_i;\dZ_i) < S(T_i;\dZ_j) \big|T_{i} < T_{j}, T_i < \Tmax, T_i\geq D_i \wedge D_j\big].
\end{equation}

As we can see in \eqref{eqn:tdu_tdc_rlt2}, the relationship depends not only on $\theta$ but also on $\tilde{\text{C}}$. $\tilde{\text{C}}$ is the concordance probability from some ``unusable'' pairs of individuals that are ignored when we estimate $\text{C}$ and $\text{C}^{td}$. Our proposed estimator $\tdu$ does not estimate $\tilde{\text{C}}$ since it only evaluates usable pairs following the rule proposed by Harrell's C-index \citep{harrell1982evaluating}. However, the influence of   $\tilde{\text{C}}$ to $\text{C}$ can be quantified from \ref{eqn:tdu_tdc_rlt2}, especially the value of $\theta$ representing the probability that the event time of individual $i$ is uncensored given that $T_i<T_j$ and $T_i<\Tmax$. 

Based on all possible values of $\theta$, the effects of $\tilde{\text{C}}$ on $\text{C}$ can be explained as follows:
\renewcommand\theenumi{\roman{enumi}}
\begin{enumerate}
	\item $\theta=1$, which means that all individuals are uncensored. This situation is where  $\tdu$ is the same as $\tdc$. Their values furthermore represent the true values of both discrimination measures. In this case, $\tilde{\text{C}}$ does not have any contribution to $\text{C}$ as $(1-\theta)=0$.
	\item $\theta \in (0,1)$. This situation is the most common case in survival analysis since survival data usually contains censored observations. In practice, $\tdu$ has different values to $\tdc$, where $\tdu$ is closer to their true values. In this case, our prediction shows that the contribution of $\tilde{\text{C}}$ is high when $\theta$ is close to zero. Otherwise, its contribution to $\text{C}$ is minimal.
	\item $\theta=0$. This case is the most extreme in survival analysis because all survival data are censored. Most survival models cannot be fitted to those datasets since they rely on at least one uncensored observation. In this case, $\text{C}$ is the same with $\tilde{\text{C}}$, which means that we cannot estimate the value of $\text{C}$ in practice. In other words, $\tdu$ is undefined.
\end{enumerate}

%---------------------------------------------
\section{Simulation Studies}\label{sec3}

\subsection{Evaluated Non-Linear Survival Models} \label{subsec:nonlinear_model}

In this section, we explore properties of Uno's C-index, Time-Depentent Concordance and the proposed model evaluation measures on simulated data. As our modelling tool, we use Nnet-survival. Since Nnet-survival is a discrete-time non-linear survival model, we need to generate discrete event and censoring times. Two standard methods for generating synthetic discrete-time survival data are: (1) simulating directly the discrete-time data using discrete-time survival models\cite{kvamme2021continuous}\cite{schmid2018discrimination} and (2) discretising the generated continuous-time survival data into a number of periods\cite{gensheimer2019scalable}. This study applies both approaches depending on the simulation setting. We perform three simulations under different settings in order to demonstrate the measures' properties. In all simulations, we generate one data set of size $n_{\text{train}}=1000$ for model training and $100$ independent test data sets also of size $n_{\text{test}}=1000$. The $100$ replications are used to produce the plots of the measures' estimators. We present our main results using boxplots, in which the horizontal lines within the boxes are medians. 
%--------------------------------------
\subsection{Simulation 1: PH Data with Various Percentages of Censoring}\label{subsec:ph_data}

In Simulation 1 we evaluate the behaviour of the discussed performance measures for a fixed model under different percentages of censored observations.
We generate right-censored continuous event times by inverting the cumulative hazard function of a CPH model\cite{bender2005generating}, where the baseline hazard rate follows Gompertz distribution with scale parameter ($\alpha$) and shape parameter ($\gamma$). In particular, the continuous event times are generated using the following formula: 
\begin{equation}\label{eqn:gompertz_time}
	T=\frac{1}{\alpha}\log\left(1-\frac{\alpha \log(v)}{\lambda\exp{(\boldsymbol{\beta} \dZ)}}\right),
\end{equation}
where $v \sim U(0,1)$, and  $\boldsymbol{\beta}'=(\beta_1,\beta_2,\beta_3,\beta_4,\beta_5)'$ is the coefficient vector of baseline covariates $\dZ=(Z_1,Z_2,Z_3,Z_4,Z_5)'$. Event times are generated based on \eqref{eqn:gompertz_time} with $\alpha=0.001$ and $\lambda=0.1$. For the CPH model, we use five covariates observed at baseline: three are generated from Bernoulli distributions ($Z_1 \sim \text{Ber}(1,0.1),~ Z_2 \sim \text{Ber}(1,0.5),~ Z_3 \sim \text{Ber}(1,0.3)$) and two from normal distributions ($Z_4 \sim \mathscr{N}(0,1), ~Z_5 \sim \mathscr{N}(0,0.5)$), where $\beta_1=3,~\beta_2=0.5,~\beta_3=0.8,~ \beta_4=0.25,$ and $\beta_5=0.95$ are their respective effect sizes. 

Follow-up time ends at $\Tmax=70$ such that all individuals with event times greater than or equal to $\Tmax$ are administratively censored. Denote by $n_{\text{surv}}$ the number of individuals surviving at least until $\Tmax$ and by $N$ the total number of simulated times. To obtain censoring within $[0,\Tmax)$, $(N-n_{\text{surv}})$ independent right censored times are generated from $\text{Weibull}$ distribution with three parameters corresponding to shape, scale, and location. The values of these parameters are fine-tuned to achieve the desired proportions of censored individuals. The $(N-n_{\text{surv}})$ observed times in $[0,\Tmax)$ are obtained by taking the minimum between the event times and the censoring times. We discretise all observed times into $11$ periods using $\mathcal{D}_1=\{0, 5, 10, 15, 20, 25, 30, 40, 50, 60, 70, \infty \}$  as the set of discretisation points. The last boundary points (i.e. $70$ and $\infty$) group all observed times greater than or equal to 70 as the last period $\Tmax$. We fit the Nnet-survival architecture (`Sim.1' column of Table~\ref{table:tab1_app} in Appendix~\ref{app2}) to a single train data set with 0\% censoring rate. Finally, the model discrimination performance is evaluated on the 100 independent test data with various censoring rates, i.e. 0\%, 4\%, 25\%, 45\%, 62\%, and 75\%.

With regards to censoring induced bias, Fig.~\ref{fig1} shows the performance measures evaluated on test samples with varying percentages of censored data. It is clear that $\tdu$ is unbiased. On the other hand, $\tdc$ is showing a positive (in the sense that they indicate better performance) bias in the presence of censoring. The higher the percentage of censored observations, the more pronounced the bias. Unsurprisingly, as the censoring rate increases, the stability of all measures decreases as shown by the increasing standard deviations. 
%-------------------------------
\begin{figure}[t]
\centering
\includegraphics[width=275pt,height=14pc]{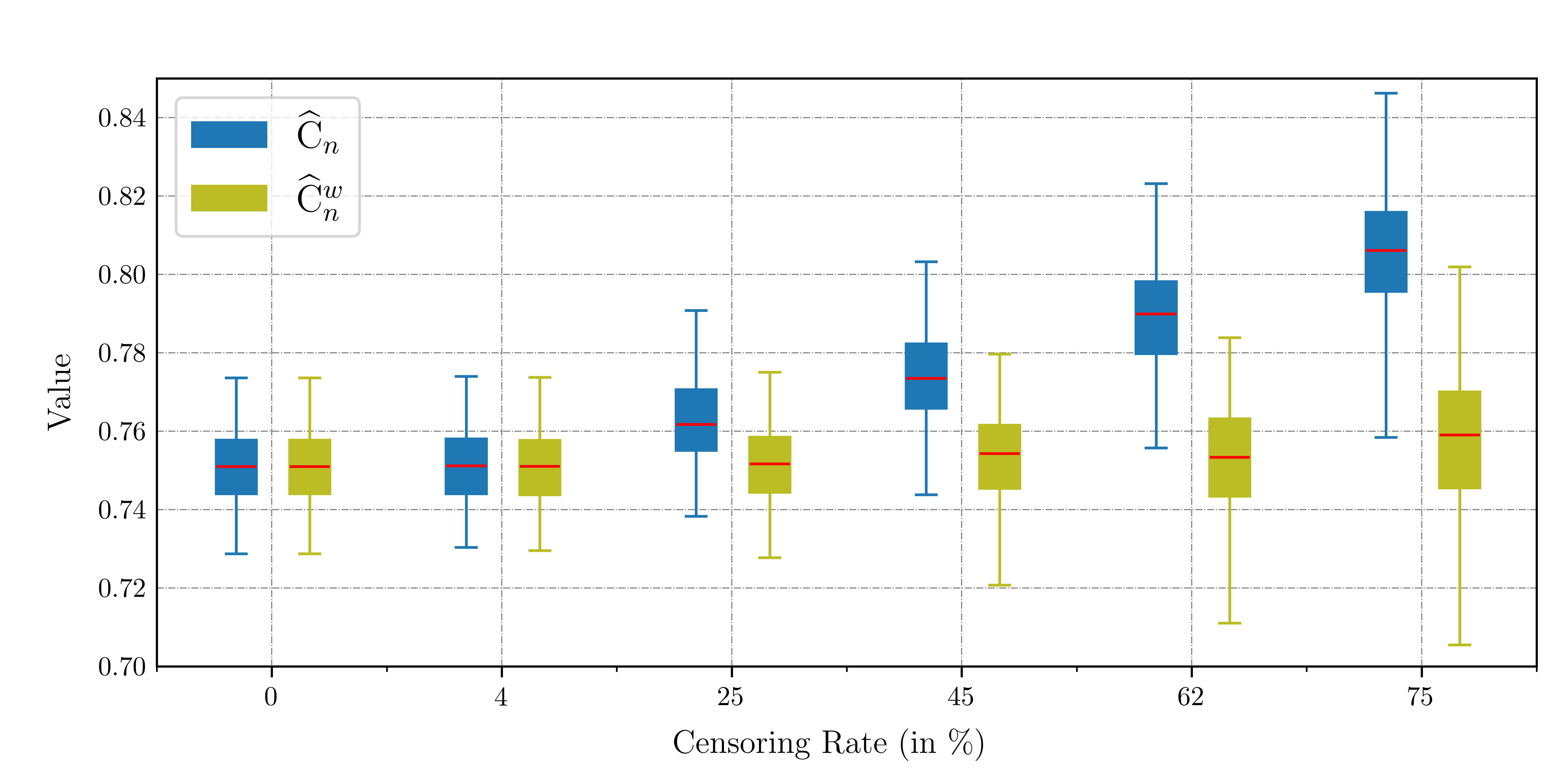}
\caption{The bias of time-dependent concordance of an Nnet-survival architecture (`Sim.1' column of Table~\ref{table:tab1_app} in Appendix~\ref{app2}) over six different censoring rates in the test data. They were estimated based on the 100 independent test data ($n_\text{test}$=1000) from a fixed fully uncensored train data ($n_\text{train}$=1000).}
\label{fig1}
\end{figure}

%----------------------------------------
\subsection{Simulation 2: Discrete-time Survival with Fixed Censoring Proportions}\label{sec:sim2}

The aims of this simulation are: (1) to demonstrate that the bias due to censoring is not always positive (meaning an improved measure), and (2) to study the effect of the estimation of $G$. There are situations where increasing censoring causes the performance indices to deteriorate. A plausible explanation, which is implemented in this experiment, is due to the model's varying discrimination ability at different event times and its interaction with the distribution of censoring times. One crucial feature of time-dependent Uno's C-index is its dependency on the estimation quality of $G$. In this experiment, we compare the effect of $G$ estimated with the KM estimator and the exact $G$ obtained from the used population distribution of censoring times on the bias of time-dependent Uno's C-index.

We generate event times based on 
\begin{equation}\label{eqn:gompertz_time2}
    T=\frac{1}{\alpha}\log\left(1-\frac{\alpha \log(v)}{\lambda \exp{(5 -4Z)}}\right),
\end{equation}
where  $v \sim U(0,1)$, $\alpha=0.0005, \lambda=0.3$, and $Z \sim \text{Ber}(1,0.5)$. In this simulation setting, we have a simpler survival model in which the event times depend only on one covariate. 
The event times are discretised into $15$ periods using $\mathcal{D}_{2}$=\{0, 4, 7, 9.5, 11.5, 13, 14, 16, 17, 19, 21, 23, 25, 28, 35, $\infty$\}. While all individuals are uncensored in the train data ($n_\text{train}$=1000), we vary the censoring rate in the test data ($n_\text{test}$=2000). For the censoring distribution, we generate discrete censored times using several predetermined probabilities in each period ($p_t)$ for $t=1,\cdots,15$ such that the desired censoring rates are achieved (see Table~\ref{table:tab2_app} in Appendix~\ref{app2}). The observed times are the minimum between the event times and the censoring times, where all individuals in the last period ($\Tmax$) are administratively censored. 

\begin{figure*}[t]
\centering
\includegraphics[width=475pt,height=16pc]{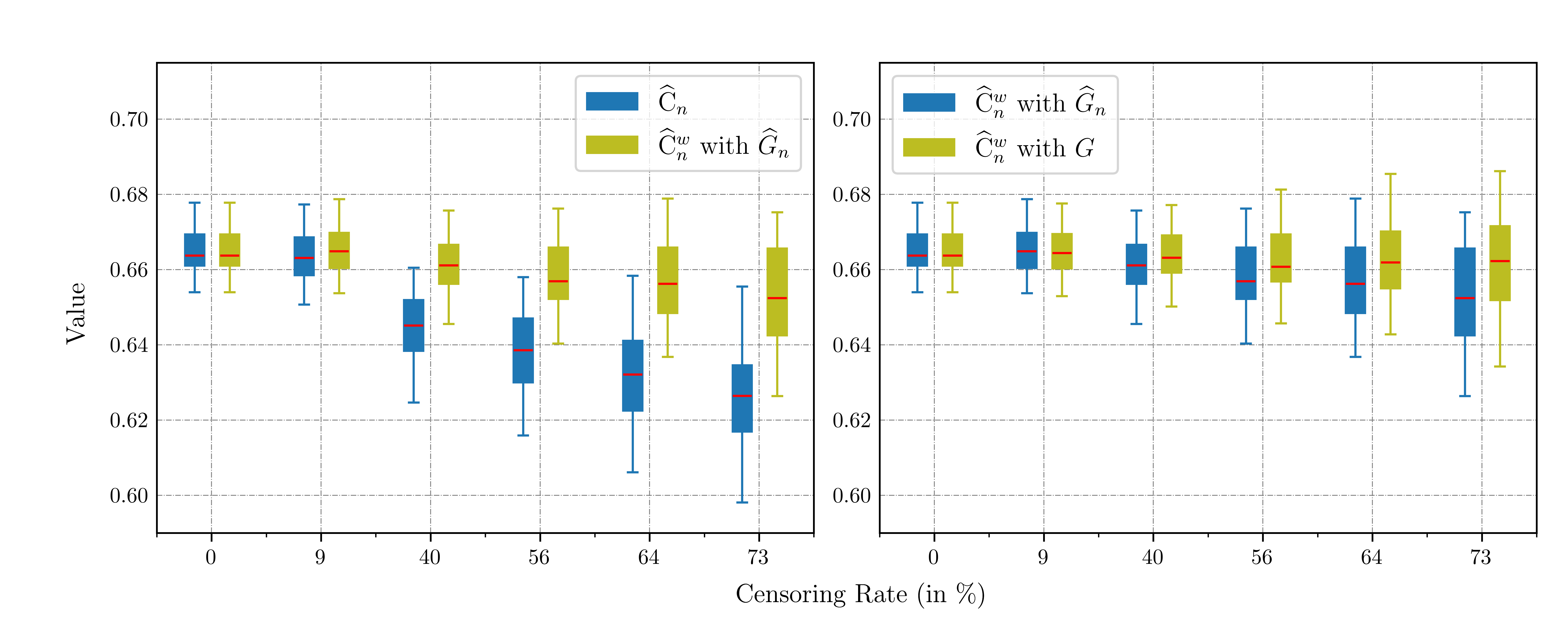}
\caption{Time-dependent concordance and time-dependent Uno's C-index of an Nnet-survival (`Sim.2' column of Table~\ref{table:tab1_app} in Appendix~\ref{app2}) over six different censoring rates in the test data. They were estimated from 100 independent test data ($n_\text{test}$=2000) from a fixed fully uncensored train data ($n_\text{train}$=1000). The figure shows the downward bias of $\tdc$ (Left) and the effect of the estimation of $G$ on $\tdu$ (Right).}
\label{fig2}
\end{figure*}

We preserve the generated train data, but in the test data, we randomly permute the event periods for a fraction of individuals whose event periods are less than or equal to $z$. We maintain the relationship between the event periods and the covariates for individuals whose event periods are greater than $z$. We choose $z=7$ for our illustration, resulting in around 70\% individuals whose event periods are randomly permuted. This setting aims to decrease the prediction accuracy of the models for shorter follow-up. The combination of the chosen $p_t$ and the covariate permutation with $z=7$ in the test data causes the downward bias of time-dependent concordance. This illustration provides a plausible explanation for the behaviour observed for heart failure data, namely downward bias of time-dependent concordance (see Section~\ref{app:nnet}).

We fit Nnet-survival architecture in `Sim.2' column of Table~\ref{table:tab1_app} in Appendix~\ref{app2} to a single train data with 0\% censoring rate. Then, we evaluate the model performance on 100 independent test data with varying censoring rates (i.e. 0\%, 9\%, 40\%, 56\%, 64\%, and 73\%). We employ two approaches in computing the censoring weights for time-dependent Uno's C-index: (1) $\widehat{G}_n$, and (2) $G$. The first approach uses the KM estimator, which is common throughout this paper, meanwhile in the second approach, we apply $G$ obtained from the population probability distribution of censoring $p_t$ for each censoring rate.

Fig.\ref{fig2} shows the downward bias of $\tdc$ (left panel) and the effect of the estimation of $G$ (right panel). As we can see in the left panel of the figure, $\tdc$ decreases as the censoring rate increases. The downward bias may happen depending on the fitted model and the prediction quality. Due to the permutation applied to some individuals in the test data, we have decreased the prediction ability of the models. On the other hand, $\tdu$ (with $\widehat{G}_n$) is much closer to its ``true'' value, i.e the value at 0\% censoring. In the right panel of Fig.\ref{fig2}, $\tdu$ with $\widehat{G}_n$ has larger bias than $\tdu$ with the true $G$. This result confirms that misspecification error in the estimation of $G$ affects the bias of IPCW measures, such as $\tdu$.

%----------------------------------------
\subsection{Simulation 3: Non PH data}\label{sec:sim3}

The main purpose of this simulation is to test the performance of Uno's C-index and time-dependent Uno's C-index in either the presence or the absence of proportional hazards. To generate non-PH data, we modify slightly the scenario in Section~\ref{subsec:ph_data}. The PH assumption is violated by assuming parameter $\alpha$ in \eqref{eqn:gompertz_time} for the event times to vary depending on the covariates. In particular, we define
\begin{equation*}
    \alpha = 
    \begin{cases}
    0.1 & \text{if $Z_4Z_5 \leq 0$ } \\
    0.4 & \text{otherwise},
    \end{cases}
\end{equation*}
where we recall that the event time depends on five covariates, i.e. $Z_1, Z_2, Z_3, Z_4$ and $Z_5$ defined in Simulation 1. Because PH assumption is usually violated in the longer follow-up, we change $\Tmax$ from 70 in Simulation 1 to $150$. For the other parameters, such as $\boldsymbol{\beta}'=(\beta_1,\beta_2,\beta_3,\beta_4,\beta_5)'$, the values were kept the same as for the PH data generation setting in Simulation 1. The observed times are discretised into $16$ periods using $\mathcal{D}_{3}$=\{0, 2, 3, 3.5, 3.75, 4, 5, 7.5, 10, 15, 20, 30, 50, 80, 90, 150, $\infty$\}. According to the conducted PH assumption Schoenfeld residuals test, non-PH was significantly present in the generated data. 

There is a single train data set $(n_{\text{train}}=1000)$ and 100 independent test data $(n_{\text{test}}=1000)$ for each data type, i.e. PH and non-PH data. We fitted Nnet-survival with the architecture given in the `Sim.1' column of Table~\ref{table:tab1_app} in Appendix~\ref{app2} to the PH data. For the non-PH data, we used Nnet-survival with the architecture described in `Sim.3' column of Table~\ref{table:tab1_app} in Appendix~\ref{app2}. Then, we computed $\uno(t)$ for each $t \in \{1,\cdots,\Tmax-1\}$ from the 100 independent test data. We only present the model evaluation results of Nnet-survival over \{1,$\cdots$,$\Tmax -1$\} periods using $\uno(t)$ since $\har(t)$ and $\uno(t)$ are exactly the same when the censoring rate is $0\%$.

We present the results of Simulation 3, where the discrimination performance of Nnet-survival at all $(\Tmax-1)$ periods based on $\uno(t)$ and $\tdu$ in the PH data (Fig.~\ref{figc1} in Appendix~\ref{app3}) and non-PH data (Fig.~\ref{fig3}) is reported, where the data generating mechanism and the discretisation points for the PH data follow the scenario in Section~\ref{subsec:ph_data}. As $\har(t)=\uno(t)$ and $\tdc=\tdu$ when censoring rate is $0\%$, we omit $\har(t)$ and $\tdc$ in those figures. As we see in Fig.~\ref{figc1}, Uno's C-index $\uno(t)$ and time-dependent Uno's C-index $\tdu$ give an almost the same result of about $75\%$. By design, $\uno(t)$ is evaluated in each period, whereas $\tdu$ gives a single value for all periods. This is expected as the simulated data are generated through and are a good fit to a proportional hazards model. PH property maintains rank of outputs such that each usable pair has constant contribution over the follow-up.

This does not remain the case for the plots of Fig.~\ref{fig3}, where the data come from a non-proportional hazard model. Although $\tdu$ gives a value of around $86\%$ (right panel), it is clear that $\uno(t)$ (left panel) gives a very different and unstable picture. Depending on which period is used to acquire the model prediction for survival probability, the index changes significantly. Note that both indices are evaluating the same model. Uno's C-index may oscillate when the PH assumption is violated due to the rank reversion. The reported model performance over the follow-up is inconsistent, resulting in the overall performance measures, e.g. $\tdu$, being preferable in such situation. 

From Fig.~\ref{fig3}, we can also see that $\tdu$ is mainly higher than $\uno(t)$ for each $t \in \{1,\cdots,\Tmax-1\}$. Because $\uno(t)$ is calculated based on usable pairs of the predicted survival curves at $t$ of all individuals, $\uno(t)$ should show false levels of accuracy when there are reversions of the predicted survival curves over follow-up. On the other hand, $\tdu$ utilises more specific usable pairs over the follow-up as given by \eqref{eqn:tdu}. As a result, $\tdu$ has predominantly a higher value than $\uno(t)$.

%----------------------------------------
\begin{figure*}[t]
\centering
\includegraphics[width=400pt,height=16pc]{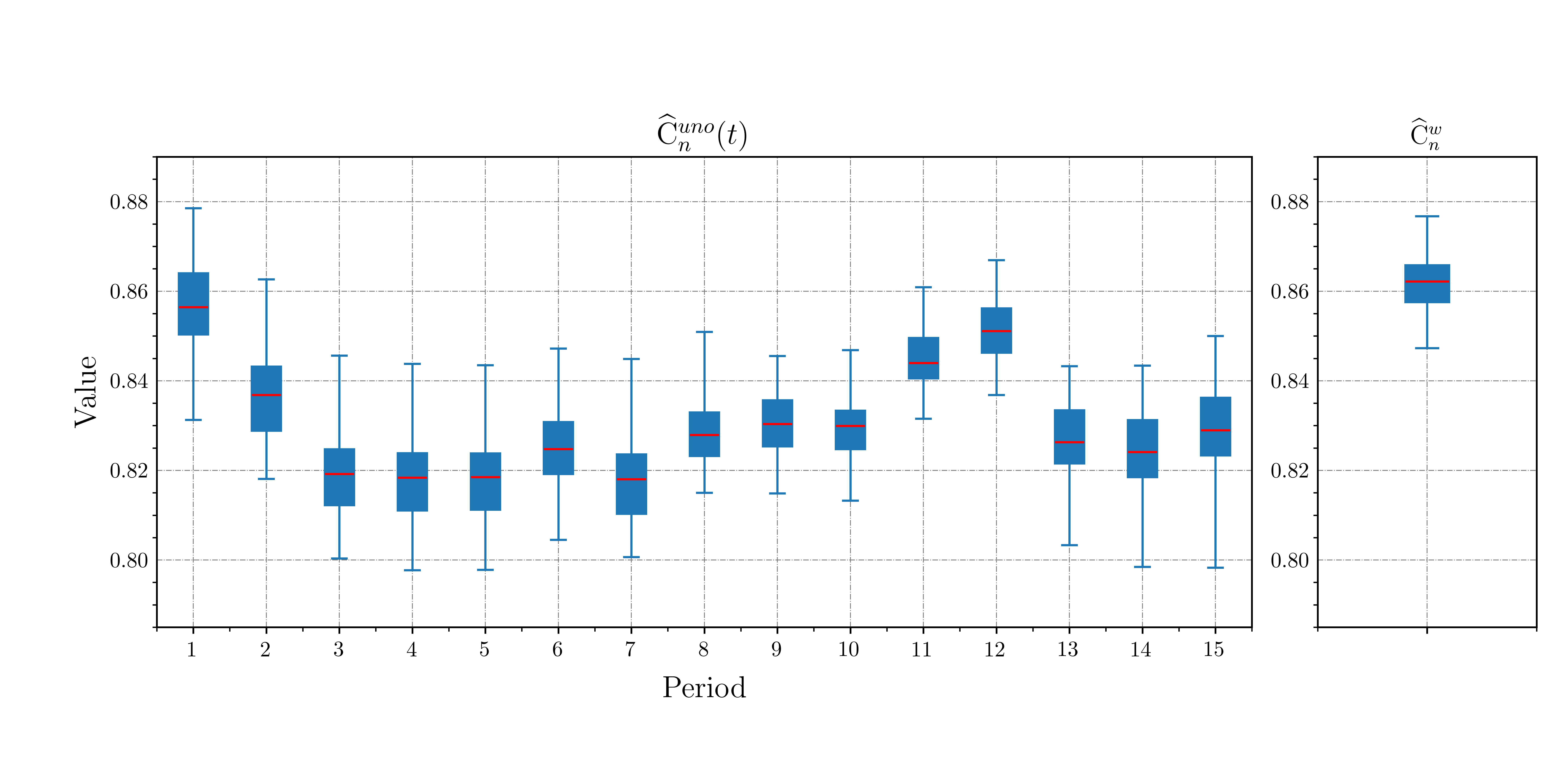}
\caption{The results of Simulation 3 aimed to show Uno's C-index (Left) and time-dependent Uno's C-index (Right) from an Nnet survival architecture (`Sim.3' column of Table~\ref{table:tab1_app} in Appendix~\ref{app2}) fitted to the non-PH data over \{1, $\cdots$, $\Tmax -1$\}. They were computed on 100 independent fully uncensored test data ($n_\text{test}$=1000) from a fixed fully uncensored train data ($n_\text{train}$=1000).}
\label{fig3}
\end{figure*}
%----------------------------------------

\begin{remark}\label{rem:nph_data}
    Although $\har(t)$ and $\uno(t)$ may be easily employed by common standard survival analysis packages either in R or Python programming, we need to carefully check whether PH assumption holds in the data. For instance, in the original paper of Nnet-survival \cite{gensheimer2019scalable}, the authors applied $\har(t)$ for assessing model performance at one period (one-year prediction) even though the PH assumption was violated in their real data.
\end{remark}

%-------------------------------------------
\section{Real Data Applications}\label{sec4}		

The primary objective of this section is to study how time-dependent Uno's C-index and time-dependent concordance behave when assessing the performance of non-linear survival models in real-world data. We focus on two practical aspects: a) the performance of the model according to each measure and b) tuning model hyper-parameters using the measures. We work with two real data sets :
\begin{enumerate}[1.]
    \item Heart failure (HF) data from Ahmad et al. \cite{ahmad2017survival} (see also Chicco and Jurman \cite{chicco2020machine} for further analysis). The HF data recorded the time to death due to heart failure of 299 patients and their baseline clinical information. The response variable is the time until the death of patients (in days), and eleven available clinical features were used as covariates (i.e. six continuous and five categorical variables). The follow-up is 289 days, during which 96 (32\%) patients died, and 203 (68\%) patients were censored. 
    \item TCGA mutation data, or simply TCGA data, was part of the TCGA project \citep{weinstein2013cancer} (see also \cite{kandoth2013mutational}). We extracted data from the R package \textbf{ dnet} \citep{fang2014thednet}. The data contains clinical information of 3,096 cancer patients (12 cancer types) with  19,428 baseline variables, including the time until the death of the patients (in days) during the study. We did feature selection using random survival forests (RSF) \citep{ishwaran2008random} and a Python package \textbf{scikit-learn} \citep{scikit-learn} with `eli5' and `PermutationImportance' functions. As a result, we only employed 13 independent variables as the model's covariates denoted by `CSMD3', `EGFR', `FLG', `MUC16', `MUC4', `PIK3CA', `PTEN', `TP53', `TTN', `USH2A', `Age', `Gender', and `TCGA tumor type'. The follow-up was defined from zero up to the largest recorded observed time (6,975 days), during which 2,099 (68\%) patients were censored, and 997 (32\%) patients died.  
\end{enumerate}

To obtain the training and test data for both datasets, we randomly divided the original data into 70\% training data and 30\% test data, where the censoring rates in the train and the test data were kept similar to the original data. We repeated the random division 100 times to get 100 pairs of independent train and test data. Then, the train and test data in HF and TCGA data were discretised using $\mathcal{D}_{4}$ = \{0, 14.25, 28.5, 42.75, 57, 71.25, 85.5, 99.75, 114, 128.25, 142.5, 156.75, 171, 185.25, 199.5, $\infty$\} and $\mathcal{D}_{5}$ = \{0, 74, 152.003, 234.465, 321.928, 415.039, 514.573, 621.488, 736.966, 862.497, 1000, $\infty$\}, respectively.

%-----------------------------------------
\subsection{Model Performance of Nnet-Survival}\label{app:nnet}

This section aims to study how $\tdu$ and $\tdc$ are evaluating the predictive performance of Nnet-survival fitted to HF and TCGA data. The used Nnet-survival architectures can be found in the columns `HF Data' and `TCGA Data' of Table~\ref{table:tab1_app}, respectively. Then, we evaluated the models' performance on their respective 100 independent test data.

Figure~\ref{fig5} shows the results of the numerical implementation. The left panel of the figure presents the values of $\tdu$ and $\tdc$ for the two datasets, where blue and yellow boxplots represent $\tdu$ and $\tdc$ from 100 test dataset , respectively. We can see from the left panel that $\tdc$ is upward (positively) biased in TCGA data. In contrast, $\tdc$ is downward (negatively) biased in HF data. To make the presentation in the left panel clearer, we also draw the boxplots of the differences between time-dependent concordance and time-dependent Uno's C-index ($\tdc-\tdu$) for each test data as given in the right panel of the figure. 

In the above, we see in practice that the non-IPCW discrimination measure, namely time-dependent concordance, can be either upward or downward biased, contrary to the positive bias (C-hacking) commonly reported in the literature (\citep{gerds2013estimating,gonen2005concordance}). We believe this behaviour is the result of a combination of different discrimination ability of the model for short and long follow-up times and the distribution of censoring times. We explored this relationship in Section \ref{sec2.3} and we demonstrated in practice that we can achieve negative bias by reducing the predictive ability of a model for short follow-up times and adjusting the censoring distribution in Simulation 2 in Section~\ref{sec3}. 

\begin{figure*}[t]
\centering
\includegraphics[width=400pt,height=14pc]{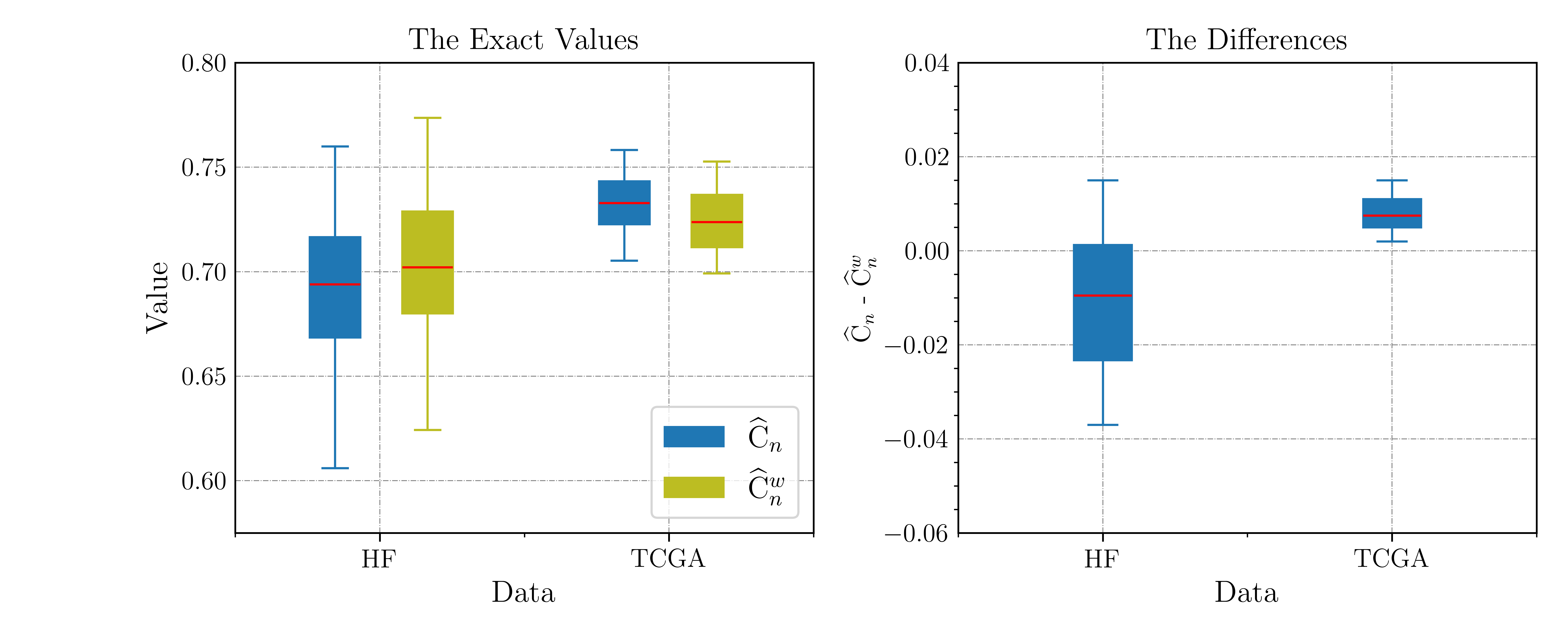}
\caption{Time-dependent concordance and time-dependent Uno's C-index (Left) and their differences (Right) of Nnet-survival architectures (`HF Data' and `TCGA Data' columns of Table~\ref{table:tab1_app} in Appendix~\ref{app2}) for HF data and TCGA data. We can clearly see the negative bias of $\tdc$ for HF data.}
\label{fig5}
\end{figure*}
%
%-----------------------------------------
\subsection{Tuning Minimum Node Size of Discrete-Time Random Survival Forests}\label{app:rsf}

All the numerical experiments discussed in this article evaluate the  performance of Nnet-survival models. Our proposed measures can be used for any non-linear survival model with a single event of interest and right-censored data and for completeness, this section shows the application of the proposed measures in evaluating the model performance of another commonly applied machine learning survival model, namely random survival forests \citep{ishwaran2008random}.

The random forests method was introduced by \cite{breiman2001random} and has been applied in vast research fields for prediction and classification tasks. \cite{ishwaran2008random} proposed random survival forests as an extension of random forests to handle right-censored survival data. Random survival forests was then adapted by \cite{schmid2020discrete} into the context of discrete-time units, where they called their approach as discrete-time survival forests (DTSF). In this work, we aim to investigate how our proposed measures behave in evaluating the model predictive performance of DTSF. In particular, we will assess the performance of DTSF when one of its hyper-parameters, namely minimum node size, varies while the other hyper-parameters are fixed.

In this implementation, we still used the same 100 pairs of train and test data in the previous as well as their discretisation setups. We first fitted the DTSF architectures in `HF Data' and `TCGA Data' columns of Table~\ref{table:tab3_app} to each train data of HF and TCGA data, respectively. Then, for each minimum node size, we evaluated the model performance from their respective test data using $\tdu$. The results for HF data can be seen in the left panel of Figure~\ref{fig6}. Due to small number of samples and high censoring rate, the variabilities of $\tdu$ each minimum node size are high. However, the results of $\tdu$ for the TCGA data are more stable, indicated by smaller standard deviations for each minimum node size (right panel of Figure~\ref{fig6}). In HF data we can pick 25 or 75 as the optimum values of the minimum node size for our DTSF architecture because they provide the largest values of $\tdu$. Meanwhile, we choose 75 as the value for the minimum node size.

\begin{figure*}
\centering
\includegraphics[width=375pt,height=14pc]{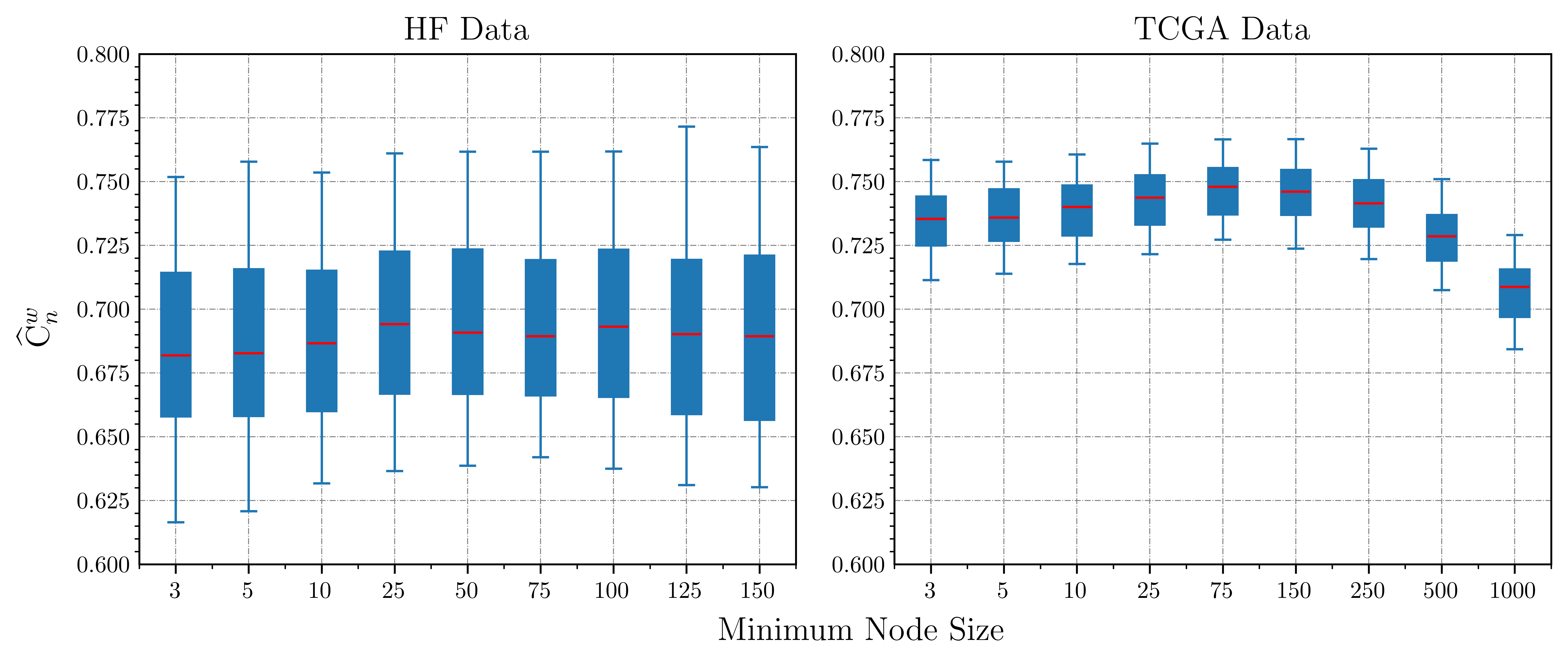}
\caption{Time-dependent Uno's C-index of a DTSF architecture (Table~\ref{table:tab3_app} in Appendix~\ref{app2}) over nine different minimum node sizes fitted to the PH data (Left) and the non-PH data (Right). They were estimated  on 100 test data ($n_\text{test}$=30\% of the data) from models fitted to 100 train data ($n_\text{train}$=70\% of the data).}
\label{fig6}
\end{figure*}

%-----------------------------------------
\section{Discussion}\label{sec5}

We discuss several discrimination measures to assess the discrimination ability of nonlinear survival models, such as Nnet-survival\cite{gensheimer2019scalable} and DTSF\cite{schmid2020discrete}. 

We first formulate the discrete time (fixed time point) versions of several well-known discrimination measures, namely Harrel's C-Index, Uno's C-Index and Antolini's Time Dependent Concordance. We discuss advantages and pitfalls of each measure, in particular with respect to a) presence of censored observations and b) violation of the PH assumption. 

In particular, we showed in our simulations (Simulation 1) how, $\tdc$ suffers from bias in the presence of censoring, being a non-IPCW measure, a disadvantage shared with Harrel's C-Index. On the other hand Uno's C-Index is unbiased in the presence of censoring but it is formulated in a way suitable for PH models only, thus providing misleading values when PH assumption is violated. This result is undesirable since it is difficult to justify the performance of the model correctly and in line with the conclusion of Sonabend et al.\cite{sonabend2022avoiding} regarding the C index. In that case, we show (Simulation 3) that the whole predicted survival curve needs to enter the model evaluation.

One possible solution to such an issue is to employ overall discrimination measures for the whole prediction horizon. In this work, we develop time-dependent Uno's C-index ($\tdu$) which can be seen as the weighted version of the time-dependent concordance ($\tdc$). Our contribution by developing a time-dependent Uno's C-index can be viewed from theoretical and practical perspective. From theoretical point of view, we provide a detailed proof of the convergence of $\tdu$ to the corresponding probability. In the convergence proof by Uno et al.\cite{uno2011c}, there is a crucial step which states that Uno's C-index is a U-statistic such that it converges almost surely to its expectation following Theorem 7 in Nolan and Pollard\cite{nolan1987u}. A common approach by a number of authors\cite{gerds2013estimating,cheung2019prioritized}, is to directly apply the results of the paper by Uno et al.\cite{uno2011c} when showing convergence of the proposed discrimination measures. We have shown how our proposed discrimination measure satisfies the assumptions of the theorem in detail, which is not trivial. What is more, we have established a necessary condition that $G(\Tmax)>\varepsilon$ for some $\varepsilon>0$ which is not mentioned in Uno et al.\cite{uno2011c}. Although limiting arguments can be provided to remove this condition, in practice it is needed to prevent the estimators from blowing up near $\Tmax$.   

We also show that time-dependent concordance can be either upward or downward biased and provide a plausible explanation for the latter. Since time-dependent concordance is an extension of Harrell's C-index, our results provide insights on how Harrell's C-index behaves. 

We also show that $\tdc$ has the same issue as $\har(t)$ because their population concordance probabilities still depend on censoring distribution. Our simulation study has confirmed such bias where $\tdc$ becomes more biased as censoring rates increase. The simulation study furthermore shows that the bias direction can be either upward or downward depending on the fitted models and the censoring distribution. On the other hand, $\tdu$ is more stable and closer to the true values. The results also explain that the bias of $\tdu$ depends highly on the estimator of $G$ although $\tdu$ is theoretically an unbiased estimator of the population concordance probability. The better $\widehat{G}$, the less bias of $\tdu$. 

In this paper, we restricted the discussions to a single event of interest and right-censored survival data. The subsequent potential research question is how we extend the time-dependent Uno's C-index to evaluate the performance of other non-linear survival models developed for survival data with different characteristics. In current literature, several machine learning approaches have been proposed for more complex cases in survival analysis, such as Neural survival recommender \citep{jing2017neural} as the long short-term networks (LSTM) for recurrent events, partial logistic artificial neural network competing risks automatic relevance determination (PLANNCR-ARD) \citep{LISBOA20031} and random survival forests competing risks (RSFCR) \citep{ishwaran2014random} for the competing risks, and neural networks for interval-censored survival data  \citep{meixide2024neural}. In these models, we cannot directly apply the time-dependent Uno's C-index to evaluate their model's performance, requiring further research and extensions.

\bmsection*{Acknowledgments}
The authors are grateful for the funding provided by
Indonesian Endowment Fund for Education (LPDP), grant ref no. S-2468/LPDP.4/2019. G. Aivaliotis  would like to acknowledge the contribution of the COST Action CA21169, supported by COST (European Cooperation in Science and Technology). This work was undertaken on ARC3 and ARC4, part of the High Performance Computing facilities at the University of Leeds, UK.

\bmsection*{Financial disclosure}

None reported.

\bmsection*{Conflict of interest}

The authors declare no potential conflict of interests.

\bibliography{wileyNJD-AMA}

\bmsection*{Supporting information}

Additional supporting information may be found in the
online version of the article at the publisher’s website.

\appendix

\bmsection{Proofs of The Convergence}\label{app1}
\vspace*{12pt}
We first present the statistical theory of U-statistics and its convergence results. In the following definition and theorem, we do not present the original versions, but we adapt and extend their notations without losing their essence so that they can fit into our discussions throughout this work.

\begin{definition}[U-statistic \citep{nolan1987u}]\label{def1}
	Let $\mathcal{x}_1,\mathcal{x}_2,\cdots$ be independent samples from a distribution $F$ on $\mathscr{X}$, and $\mathscr{H}$  be a class of symmetric functions on $\mathscr{X} \otimes \mathscr{X}$.
	For $h \in \mathscr{H}$, we define 
	\begin{equation}\label{eqn:sn}
		\zeta_n(h)= \sum_{1\leq i\ne j\leq n}h(\mathcal{x}_i,\mathcal{x}_j).
	\end{equation}
	Then $\frac{1}{n(n-1)}\zeta_n(h)$, indexed by $h \in \mathscr{H}$, is called the U-statistic.
\end{definition}

\begin{theorem}[\cite{serfling1980approximation} Theorem A, p.~190]\label{theo:conv_ustat_serf}
    Let $\zeta_n(h)\big/ n(n-1)$ be the U-statistic for a symmetric kernel $h \in \mathscr{H}$. Let 	
    \begin{equation*}\label{eqn:fof}
	F \otimes F(h)=\iint h(\mathcal{x}_i,\mathcal{x}_j)\,dF(\mathcal{x}_i)\,dF(\mathcal{x}_j) 
    \end{equation*}
    be the expected value of $h(\mathcal{x}_i,\mathcal{x}_j)$ with respect to $F \otimes F$.
    If $F \otimes F(\left|h\right|) < \infty$, then $\frac{\zeta_n(h)}{n(n-1)}  \overset{a.s.}{\to}  F \otimes F(h)$.
\end{theorem}

%-----------------------------------------
\bmsubsection{Convergence of Time-Dependent Concordance \label{app1.1a}}

\begin{proof}[Proof of Lemma~\ref{lem1}]
    We rewrite  the numerator of \eqref{eqn:est_tdc} as follows
    \begin{equation}\label{eqn:a1}
        \sum_{i\neq j}^n\indd{S(T_i;\dZ_i)<S(T_i;\dZ_j)}\indd{T_i<T_j, T_i < \Tmax}\indd{T_i < D_i}\indd{T_i < D_j}.
    \end{equation}
    We then denote $\mathcal{x}_i=(\dZ_i,T_i,D_i)$ and $\mathcal{x}_j=(\dZ_j,T_j,D_j)$ so that $\{\mathcal{x}_1,\mathcal{x}_2,\cdots\}$ are independent samples from distribution $F$ on $\mathscr{X}$. Let $\mathscr{H}$ be a class of functions on $\mathscr{X} \otimes \mathscr{X}$ given by
    \begin{equation*}\label{eqn:h}
        h(\mathcal{x}_i,\mathcal{x}_j)=\indd{S(T_i;\dZ_i)<S(T_i;\dZ_j)}\indd{T_i<T_j, T_i < \Tmax}\indd{T_i < D_i}\indd{T_i < D_j}.
    \end{equation*}
    Because $h$ is not symmetric, we need to symmetrise it as follows 
    \begin{equation*}\label{eqn:h_sim}
        \bar{h}\left(\mathcal{x}_i,\mathcal{x}_j\right)=\frac{1}{2}\left(h(\mathcal{x}_i,\mathcal{x}_j)+h(\mathcal{x}_j,\mathcal{x}_i)\right).
    \end{equation*}
    Then \eqref{eqn:a1} equals
    \begin{equation*}\label{eqn:u-stat2}
        \frac{2}{n(n-1)}\sum_{1\leq i<j\leq n}\bar{h}(\mathcal{x}_i,\mathcal{x}_j)
        =
        \frac{1}{n(n-1)}\sum_{1\leq i \ne j\leq n}\bar{h}(\mathcal{x}_i,\mathcal{x}_j)
    \end{equation*}
    and this expression is a $U$-statistic with kernel $\bar{h}$, see Definition~\ref{def1}.
    We have $F \otimes F(|\bar{h}|) \le 1$ since $\bar h$ takes values $0$ or $1$.
 Therefore, the condition of Theorem~\ref{theo:conv_ustat_serf}  is satisfied so that
    \begin{equation}\label{eqn:har_num_conv}
        \frac{2}{n(n-1)}\sum_{1\leq i<j\leq n}\bar{h}(\mathcal{x}_i,\mathcal{x}_j) \overset{a.s.}{\to} F \otimes F(\bar{h}).
    \end{equation}
We evaluate the right-hand side
\begin{equation*}
\begin{aligned}
\lefteqn{F \otimes F(\bar{h})}\\
	&=
        \frac{1}{2} \left[F \otimes F\big(\indd{S(T_i;\dZ_i<S(T_i;\dZ_j)}\indd{T_i<T_j, T_i < \Tmax}\indd{T_i < D_i}\indd{T_i < D_j}\big) + F \otimes F\big(\indd{S(T_j;\dZ_j)<S(T_j;\dZ_i)}\indd{T_j<T_i, T_j < \Tmax}\indd{T_j< D_j}\indd{T_j < D_i}\big)\right]\\
	&=
        \frac{1}{2} \left[\prob \big[ S(T_i;\dZ_i)<S(T_i;\dZ_j), T_i<T_j, T_i < \Tmax,T_i < D_i \wedge D_j\big] + \prob \big[ S(T_j;\dZ_j)<S(T_j;\dZ_i), T_j<T_i, T_j < \Tmax, T_j < D_j \wedge D_i\big]\right]\\
    &= \prob \big[ S(T_i;\dZ_i)<S(T_i;\dZ_j), T_i<T_j, T_i < \Tmax,T_i < D_i \wedge D_j\big],
\end{aligned}
\end{equation*}
where the last equality holds as $(\dZ_i, T_i, D_i)$ and $(\dZ_j, T_j, D_j)$ are independent and identically distributed.

We now rewrite the denominator of \eqref{eqn:est_tdc} as follows
	\begin{equation*}
		\sum_{i\neq j}^n\indd{T_i<T_j, T_i < \Tmax}\indd{T_i < D_i}\indd{T_i < D_j}.
	\end{equation*}
We notice that the above sum equals
\[
\frac{2}{n(n-1)}\sum_{1\leq i<j\leq n}\bar{\psi}(\mathcal{x}_i,\mathcal{x}_j)
\]
where $\bar{\psi}$ is the symmetrised version of $\psi(\mathcal{x}_i,\mathcal{x}_j)=\indd{T_i<T_j, T_i < \Tmax}\indd{T_i < D_i}\indd{T_i < D_j}$. Because the only difference between the numerator and the denominator is the indicator function $\indd{S(t;\dZ_i)<S(t;\dZ_j)}$ whose values is either $0$ or $1$, we can use the same approach and arguments used in obtaining the convergence in \eqref{eqn:har_num_conv} to show
\begin{equation}\label{eqn:har_den_conv}
		\frac{2}{n(n-1)}\sum_{1\leq i<j\leq n}\bar{\psi}(\mathcal{x}_i,\mathcal{x}_j)
		\overset{a.s.}{\to}
		F \otimes F(\bar{\psi}),
\end{equation}
and
\begin{equation*}
F \otimes F(\bar{\psi}) 
= \frac{1}{2} \left(\prob \big[ T_i<T_j, T_i < \Tmax, T_i < D_i \wedge D_j\big]+ \prob \big[T_j<T_i, T_j < \Tmax,  T_j < D_i \wedge D_j\big]\right)
= \prob \big[ T_i<T_j, T_i < \Tmax, T_i < D_i \wedge D_j\big].
\end{equation*}
By an application of the Continuous Mapping Theorem \citep{shao2003mathematical}, we have  
    \begin{equation*}
        \begin{aligned}
        \frac{\sum_{1 \le i < j\leq n} \bar{h}(\mathcal{x}_i,\mathcal{x}_j)}{\sum_{1 \le i< j\leq n}\bar{\psi}(\mathcal{x}_i,\mathcal{x}_j)}
	\overset{a.s.}{\to} 
	\frac{ F\otimes F(\bar{h})}{\ F\otimes F(\bar{\psi})} 
	&=
 \frac{\prob \big[ S(T_i;\dZ_i)<S(T_i;\dZ_j), T_i<T_j, T_i < \Tmax,T_i < D_i \wedge D_j\big]}{\prob \big[ T_i<T_j, T_i < \Tmax, T_i < D_i \wedge D_j\big]}\\
	&=
        \prob\left[S(T_i;\dZ_i) < S(T_i;\dZ_j) \big|T_i < T_j, T_i < \Tmax, T_i < D_i \wedge D_j\right],
	\end{aligned}
    \end{equation*}
    which concludes the convergence of $\tdc$.
\end{proof}

%-----------------------------------------
\bmsubsection{Convergence of Time-Dependent Uno's C-Index \label{app1.1b}}

\begin{definition}[Envelope \citep{nolan1987u}]\label{def2}
	Consider the class of symmetric functions $\mathscr{H}$. If $H(\cdot,\cdot) \geq |h(\cdot,\cdot)|$ for each $h \in \mathscr{H}$, then $H$ is a positive envelope for $\mathscr{H}$.
\end{definition}

\begin{definition}[Covering Number \citep{nolan1987u}]\label{def3}
	Let $H$ be a positive envelope of the class $\mathscr{H}$. For $\delta > 0$, the covering number $N_p(\delta,Q,\mathscr{H},H)$ with respect to measure $Q$ such that $0<Q(H^p)<\infty$ is defined as the smallest cardinality for a subclass $\mathscr{H}^*$ of  $\mathscr{H}$ such that
	\begin{equation*}
		\min_{h^* \in \mathscr{H}^*} Q\big|h-h^*\big|^p \leq \delta Q(H^p),  \hspace*{3pt} \text{for each}  \hspace*{3pt} h \hspace*{3pt} \text{in}  \hspace*{3pt} \mathscr{H}.
	\end{equation*}
\end{definition}

%-----------------------------------
\begin{theorem}[Uniform Almost-Sure Convergence of U-processes (see Theorem 7 in page 787 in the paper by Nolan and Pollard \cite{nolan1987u})]\label{theo:conv_ustat}
    Let $\T_n(\cdot)$ be defined as
    \begin{equation}\label{eqn:Tn_ori}
        \T_n(h)=\sum_{1\leq i\neq j\leq n} \big(h(\mathcal{x}_{2i},\mathcal{x}_{2j})+h(\mathcal{x}_{2i},\mathcal{x}_{2j-1})+h(\mathcal{x}_{2i-1},\mathcal{x}_{2j})+h(\mathcal{x}_{2i-1},\mathcal{x}_{2j-1})\big),
    \end{equation}
    where $\mathcal{x}_1,\cdots,\mathcal{x}_{2n}$ are obtained by taking a double sample from a distribution $F$ on $\mathscr{X}$, and $h$ is a function in the symmetric class $\mathscr{H}$ with envelope $H$, respectively. For any $h$ in $\mathscr{H}$, we define
    \begin{eqnarray*}
		F_n \otimes F (h)
		&=&
		\frac1n\sum_{i=1}^{n} \int h(\mathcal{x}_i, \mathcal{x}_j)\,dF(\mathcal{x}_j),
    \end{eqnarray*}
where $F_n$ is the empirical distribution of the random sample of size $n$ from $F$ (i.e., it is a random distribution),

and 
    \begin{eqnarray*}
		F \otimes F(h)
		&=&
		\iint h(\mathcal{x}_i,\mathcal{x}_j)\,dF(\mathcal{x}_i)\,dF(\mathcal{x}_j) 
    \end{eqnarray*}
    is the expected value of $h(\mathcal{x}_i,\mathcal{x}_j)$ with respect to $F \otimes F$.	
    If for each $\delta>0$,
    \begin{enumerate}[i)]
		\item $\log N_1(\delta,\T_n, \mathscr{H}, H)=o_p(n)$
		\item $\log N_1(\delta, F_n\otimes F, \mathscr{H}, H)=o_p(n)$
		\item $N_1(\delta, F\otimes F, \mathscr{H}, H)<\infty$,
    \end{enumerate}
	then 
	\begin{equation*}
		\text{sup}_{h \in \mathscr{H}}\bigg|\bigg(\frac{\zeta_n(h)}{n(n-1)}\bigg)-F\otimes F(h)\bigg| \overset{a.s.}{\to}  0,
	\end{equation*}
	as $n \to \infty$, where a.s. stands for almost sure convergence, and
	\begin{equation*}
		\zeta_n(h)= \sum_{1\leq i<j\leq n}h(\mathcal{x}_i,\mathcal{x}_j)
	\end{equation*}
	 as defined in \eqref{eqn:sn}.
\end{theorem}

The notation $o_p$ in Theorem~\ref{theo:conv_ustat} has the following meaning: when we write that $Y_n = o_p(n)$, we mean that $Y_n/n$ converges to $0$ in probability, i.e.,
\[
	\lim_{n\to\infty} \prob\left[\left|Y_n/n\right| \ge \epsilon \right] =0, \qquad \forall \epsilon > 0.
\]
In our proofs we will establish a stronger fact, that is, the covering numbers in (i) and (ii) are $o(n)$, so $Y_n(\omega)/n \to 0$ uniformly in $\omega$.
We remark that almost sure convergence implies the convergence in probability - we will use this fact in our proofs.

%------------------------------------- 
Next, we will demonstrate the convergence of the numerator of $\tdu$ by showing that all the assumptions of Theorem~\ref{theo:conv_ustat} are satisfied. In particular, the assumptions will be satisfied by the U-statistic with symmetrised kernel obtained from the  numerator of $\tdu$.

Recall the definition of $\mathcal{G}_\epsilon$ from \eqref{eqn:G_epsilon}. For $i=1,2,\cdots$, we denote by $\mathcal{x}_i=(\dZ_i,T_i,D_i)$ independent samples from distribution $F$ on $\mathscr{X}$. Define
\begin{equation}\label{eqn:hg_}
 		h^g(\mathcal{x}_i,\mathcal{x}_j)=\indd{T_i < D_i}\indd{S(T_i;\dZ_i)<S(T_i;\dZ_j)}\indd{T_i<X_j}g^{-2}(T_i), \qquad g \in \mathcal{G}_\epsilon.
\end{equation}
Let
\begin{equation}\label{eqn:h_eps}
 		\mathscr{H}_{\epsilon}=\{h^g: g \in \mathcal{G}_{\epsilon}\}, 
\end{equation}
and 
\begin{equation}\label{eqn:bar_h_eps}
    \bar{\mathscr{H}}_{\epsilon}= \{ \bar{h}^g: \bar{h}^g(x_i, x_j) = (h^g(x_i, x_j) + h^g(x_j, x_i))/2: h^g \in  \mathscr{H}_{\epsilon} \},
 \end{equation}
where $\bar{h}^{g}$ is the symmetrised version of $h^g$.

%----------------------------
\begin{lemma}[Convergence of The Symmetrised Numerator of Time-dependent Uno's C-index]\label{lemma:conv_num_tdu}
Assume that $\epsilon > 0$ and conditions (R1-R2) hold. Then
 	\begin{equation*}
 		\displaystyle \text{sup}_{\bar{h}^g\in \bar{\mathscr{H}}_{\epsilon}} \bigg| \frac{2}{n(n-1)} \sum_{1\leq i<j\leq n}\bar{h}^g(\mathcal{x}_i,\mathcal{x}_j) - F\otimes F(\bar{h}^g)\bigg| \overset{a.s.}{\to} 0,
 	\end{equation*}
as $n \to \infty$, where we use the notation introduced above.  
\end{lemma}
 
 %----------------------------
 \begin{proof}[Proof of Lemma~\ref{lemma:conv_num_tdu}]
 	
Since
\begin{equation*}
	\frac{2}{n(n-1)} \sum_{1\leq i<j\leq n}\bar{h}^g(\mathcal{x}_i,\mathcal{x}_j) 
\end{equation*}
is a U-statistic with kernel $\bar{h}^g$ (see Definition~\ref{def1}) we only need to show that $\bar{\mathscr{H}}_{\epsilon}$ 
satisfies all the assumptions of Theorem~\ref{theo:conv_ustat}. Conditions (i) and (ii) of the theorem will be shown to be satisfied for each $\omega$, so we will argue for arbitrary deterministic sequences of points in $\mathscr{X}$ instead of random samples from $F$.

We first find a positive envelope of $\bar{\mathscr{H}}_\epsilon$. As $\epsilon >0$ is the minimum value of $g \in \mathcal{G}_\epsilon$, it is easy to see that the constant function $\bar{H}_{\epsilon}=(1/\epsilon^2)$ is the sought envelope.

Recall the operator $\mathbb{T}_n$ defined in \eqref{eqn:Tn_ori}.
 	
We trivially have $\T_n(\bar{H}^{g})=4n(n-1)/\epsilon^2$. To find the upper bound of $N_1( \delta, \T_n, \bar{\mathscr{H}}_{\epsilon}, \bar{H}^{g})$, for any sequence $(x_1, \ldots, x_{2n}) \in \mathscr{X}$, we will determine a relatively small number of functions
    \begin{equation*}
 		\{\bar{h}^{g_{1}},\bar{h}^{g_{2}},\cdots, \bar{h}^{g_{k}}\} \subseteq \bar{\mathscr{H}}_{\epsilon}
    \end{equation*}
such that for any $\bar{h}^{g} \in \bar{\mathscr{H}}_{\epsilon}$ we have
\begin{equation}\label{min_Tn}
\min_{1 \leq k \leq K} \T_n\left(\big|\bar{h}^{g}-\bar{h}^{g_{k}}\big|\right) \leq 
 			\delta \T_n(\bar{H}_{\epsilon}) = (\delta/\epsilon^2) 4n(n-1).
\end{equation}

Fix a sequence $(x_1, \ldots, x_{2n}) \in \mathscr{X}$. Denote
\[
\mathscr{K}(x_i, x_j) = \indd{T_i < D_i}\indd{S(T_i;\dZ_i)<S(T_i;\dZ_j)}\indd{T_i<X_j}.
\]
For $\bar h^g \in \bar{\mathscr{H}}_\epsilon$ and $1 \leq k \leq K$, we have
    \begin{equation*}
 	\begin{aligned}
 	  &\T_n\left( \left|\bar{h}^g-\bar{h}^{g_k}\right|\right)\\
        &=  \sum_{i\neq j}^n\Big[\left|\bar{h}^g\left(x_{2 i}, x_{2 j}\right)-\bar{h}^{g_k}\left(x_{2 i}, x_{2 j}\right)\right| +\left|\bar{h}^g\left(x_{2 i}, x_{2 j-1}\right)-\bar{h}^{g_k}\left(x_{2 i}, x_{2 j-1}\right)\right|\\
        &\hspace{38pt}+ \big| \bar{h}^g\left(x_{2 i-1}, x_{2 j}\right) -\bar{h}^{g_k}\left(x_{2 i-1}, x_{2 j}\right) \big|+ \big| \bar{h}^g\left(x_{2 i-1}, x_{2 j-1}\right)-\bar{h}^{g_k}\left(x_{2 i-1}, x_{2 j-1}\right) \big| \Big] \\
 	&=
        \sum_{i\neq j}^n \frac{1}{2}\bigg[\big| \mathscr{K}\left(x_{2 i}, x_{2 j}\right) g^{-2}\left(T_{2 i}\right)+\mathscr{K}\left(x_{2 j}, x_{2 i}\right) g^{-2}\left(T_{2 j}\right) -\mathscr{K}\left(x_{2 i}, x_{2 j}\right) g_k^{-2}\left(T_{2 i}\right)-\mathscr{K}\left(x_{2 j}, x_{2 i}\right) g_k^{-2}\left(T_{2 j}\right) \big| \\
        &\hspace{38pt}+\big|\mathscr{K}\left(x_{2 i}, x_{2 j-1}\right) g^{-2}\left(T_{2 i}\right)+\mathscr{K}\left(x_{2 j-1}, x_{2 i}\right) g^{-2}\left(T_{2 j-1}\right) -\mathscr{K}\left(x_{2 i}, x_{2 j-1}\right) g_k^{-2}\left(T_{2 i}\right)-\mathscr{K}\left(x_{2 j-1}, x_{2 i}\right) g_k^{-2}\left(T_{2 j-1}\right) \big|\\
        &\hspace{38pt}+ \big| \mathscr{K}\left(x_{2 i-1}, x_{2 j}\right) g^{-2}\left(T_{2 i-1}\right)+\mathscr{K}\left(x_{2 j}, x_{2 i-1}\right) g^{-2}\left(T_{2 j}\right)- \mathscr{K}\left(x_{2 i-1}, x_{2 j}\right) g_k^{-2}\left(T_{2 i-1}\right)-\mathscr{K}\left(x_{2 j}, x_{2 i-1}\right) g_k^{-2}\left(T_{2 j}\right) \big| \\
 	&\hspace{38pt} + \big| \mathscr{K}\left(x_{2 i-1}, x_{2 j-1}\right) g^{-2}\left(T_{2 i-1}\right)+\mathscr{K}\left(x_{2 j-1}, x_{2 i-1}\right) g^{-2}\left(T_{2 j-1}\right) \\
 	&\hspace{38pt} -\mathscr{K}\left(x_{2 i-1}, x_{2 j-1}\right) g_k^{-2}\left(T_{2 i-1}\right)-\mathscr{K}\left(x_{2 j-1}, x_{2 i-1}\right) g_k^{-2}\left(T_{2 j-1}\right) \big| \bigg]\\
 	&\leq \sum_{i\neq j}^n \frac{1}{2}\Big[\left|g^{-2}\left(T_{2 i}\right)-g_k^{-2}\left(T_{2 i}\right)\right|+\left|g^{-2}\left(T_{2 j}\right)-g_k^{-2}\left(T_{2 j}\right)\right|+ \left|g^{-2}\left(T_{2 i}\right)-g_k^{-2}\left(T_{2 i}\right)\right|+\left|g^{-2}\left(T_{2 j-1}\right)-g_k^{-2}\left(T_{2 j-1}\right)\right| \\
 	&\hspace{38pt}+ \left|g^{-2}\left(T_{2 i-1}\right)-g_k^{-2}\left(T_{2 i-1}\right)\right|+\left|g^{-2}\left(T_{2 j}\right)-g_k^{-2}\left(T_{2 j}\right)\right|+\left|g^{-2}\left(T_{2 i-1}\right)-g_k^{-2}\left(T_{2 i-1}\right)\right|+\left|g^{-2}\left(T_{2 j-1}\right)-g_k^{-2}\left(T_{2 j-1}\right)\right|\Big] \\
 &= \sum_{i=1}^{2 n}(n-1)\left|g^{-2}\left(T_i\right)-g_k^{-2}\left(T_i\right)\right|+\sum_{j=1}^{2 n}(n-1)\left|g^{-2}\left(T_j\right)-g_k^{-2}\left(T_j\right)\right|\\
 & = 2(n-1) \sum_{i=1}^{2 n}\left|g^{-2}\left(T_i\right)-g_k^{-2}\left(T_i\right)\right|,
        \end{aligned}
    \end{equation*}
where the inequality is due to the triangle inequality and $0 \leq \mathscr{K}(\cdot,\cdot) \leq 1 $. Therefore, 
\begin{align*}
\min_{1 \leq k \leq K} \T_n\big|\bar{h}^{g}-\bar{h}^{g_{k}}\big|
&\leq
 		\min_{1 \leq k \leq K}\left(2(n-1) \sum_{i=1}^{2 n}\left|g^{-2}\left(T_i\right)-g_k^{-2}\left(T_i\right)\right|\right)\\
&\leq
2(n-1) \min_{1 \leq k \leq K} \sum_{i=1}^{2n} \max_{1 \leq r \leq 2n}\big|g^{-2}(T_{r})-g_{k}^{-2}(T_r)\big|\\
&=
4n(n-1) \min_{1 \leq k \leq K} \max_{1 \leq r \leq 2n}  \big|g^{-2}(T_r)-g_{k}^{-2}(T_r)\big|.
\end{align*}

Thus, \eqref{min_Tn} holds if
\begin{equation}\label{minmax_gn}
\min_{1 \leq k \leq K}  \max_{1 \leq r \leq 2n} \left(  \big|g^{-2}(T_r)-g_{k}^{-2}(T_r)\big|\right) \leq \delta^*,
\end{equation}
where $\delta^*=(\delta/\epsilon^2)$.

For convenience of arguments, let $w := g^{-2}$ so that 
\begin{equation}\label{W}
\begin{aligned}
\mathscr{W}_\epsilon&= 
\left\{w:\ w = g^{-2} \text { for } g \in \mathcal{G}_{\epsilon}\right\}
= 
\big\{w: [0,\Tmax)\to \left[1, 1 / \epsilon^2\right]:\ \text{$w$ non-decreasing and right-continuous}\big\} .
\end{aligned}
\end{equation} 	
Our task is to find a subset $\mathscr{W}^*=\{w_1,w_2,\cdots,w_K\}$ of $\mathscr{W}_\epsilon$ such that
\begin{equation}\label{eqn:ww}
 \min_{k=1, \ldots, K} \max_{1\leq r \leq 2n}\big|w(T_r)-w_k(T_r)\big|\leq \delta^*,
\end{equation}
To satisfy assumption (i) of Theorem \ref{theo:conv_ustat}, we want to have $\log(K) = o(n)$.

We start by splitting the interval $[1,(1/\epsilon^2)]$ into  $M_{\delta^*}=(1/\epsilon^2)/\delta^*$ sub-intervals and denote the middle points of those intervals by $z_l$, $l=1, \ldots, M_{\delta^*}$, with $z_1 < z_2 < \cdots < z_{M_{\delta^*}}$. Notice that the width of each interval is less than $\delta$. For the convenience of presentation, assume that $T_1 < T_2 < \cdots < T_{2n}$. Notice that the values of functions in $\mathscr{W}^*$ outside of $\{T_1, \ldots, T_{2n}\}$ are irrelevant, so we will determine those values on this grid of times and extend to $[0, T_{max})$ in a any non-decreasing way. Consider a sequence of indices $1 \le l_1 \le l_2 \le \cdots \le l_{2n} \le M_{\delta^*}$. This sequence of indices defines a function $w_{l_1, \ldots, l_{2n}}$ on the grid $T_1, \ldots, T_{2n}$ by
\[
w_{l_1, \ldots, l_{2n}} (T_i) = z_i, \qquad i=1, \ldots, 2n.
\]
These functions are non-decreasing and, due to the construction of the grid $(z_l)$, the inequality \eqref{eqn:ww} is satisfied. Thanks to the monotonicity of indices, the number of those functions does not grow exponentially in $n$. To the contrary, the number of functions (or the number of non-decreasing sequences of indices) is equal the number of configurations of $2n$ balls in $M_\delta$ bins. The latter is bounded from above by $(2n+1)^{M_\delta}$ as in each bin (when seen separately from others) there may be $0,1, \ldots, 2n$ balls.

Summarising, we have shown how to construct a family $\mathscr{W}^*$ of size $K$ that satisfies \eqref{eqn:ww} and $K \le (2n+1)^{M_{\delta^*}}$.

Hence,
    \begin{equation}\label{up_bound}
 		N_1( \delta, \T_n, \bar{\mathscr{H}}_{\epsilon}, \bar{H}_{\epsilon}) \leq (2n+1)^{M_{\delta^*}}.
    \end{equation}
Taking the logarithm of the bound and dividing by $n$ given the limit
    \begin{equation}\label{lim_up_bound}
        \lim_{n\to\infty}  \frac{\log\big((2n+1)^{M_{\delta^*}}\big)}{n} =\lim_{n\to\infty}  \frac{M_{\delta^*} \log(2n+1)}{n} = \lim_{n\to\infty} \frac{M_{\delta^*} }{2n+1} = 0,
    \end{equation}
where the penultimate equality is due to L'Hospital's rule. Hence, assumption (i) of Theorem \ref{theo:conv_ustat} is satisfied.

In order to estimate $N_1(\delta, F_n \otimes F, \bar{\mathscr{H}}_{\epsilon}, \bar{H}^{g})$, we notice that
\begin{equation*}
\bar{\mathscr{H}}_{\epsilon} = \mathscr{V}^1_\epsilon + \mathscr{V}^2_{\epsilon}:=\big\{V_1+V_2: V_1^{g} \in \mathscr{V}^1_{\epsilon}, V_2 \in \mathscr{V}^2_{\epsilon} \big\},
\end{equation*}
where
\begin{equation*}
\mathscr{V}^1_\epsilon = \left\{ V_1^g:\ V_1^g(\mathcal{x}_i,\mathcal{x}_j) =\frac{1}{2}\mathscr{K}(\mathcal{x}_i,\mathcal{x}_j) g^{-2}(\mathcal{x}_i),\ g \in \mathcal{G}_\epsilon \right\}
\end{equation*}
   and
\begin{equation*}
\mathscr{V}^2_\epsilon = \left\{ V_2^g:\ V_2^g(\mathcal{x}_i,\mathcal{x}_j) =\frac{1}{2}\mathscr{K}(\mathcal{x}_j,\mathcal{x}_i) g^{-2}(\mathcal{x}_j),\ g \in \mathcal{G}_\epsilon \right\}.
\end{equation*}
We further have\cite[Lemma~16]{nolan1987u}
\begin{equation}\label{eqn:sum_covering}
N_1(\delta, F_n \otimes F, \bar{\mathscr{H}}_{\epsilon}, \bar{H}_{\epsilon})
\leq
N_1(\delta/4, F_n \otimes F, \mathscr{V}^1_{\epsilon}, \bar{H}_{\epsilon})\ N_1(\delta/4, F_n \otimes F, \mathscr{V}^2_\epsilon, \bar{H}_{\epsilon})
\end{equation}
upon noticing that $\bar{H}_{\epsilon}$ is also an envelope of $\mathscr{V}^1_{\epsilon}$ and $\mathscr{V}^2_{\epsilon}$. Since $\mathcal{K}$ takes values $0$ or $1$, we have
\[
N_1(\delta/4, F_n \otimes F, \mathscr{V}^1_{\epsilon}, \bar{H}_{\epsilon}) \le N_1(\delta/4, 1-G_n, \mathscr{W}_{\epsilon}, \bar{W}_{\epsilon})
\]
with $\bar W_\epsilon = 1/\epsilon^2$ and $G_n$ denoting the empirical tail distribution function of $T_i$, the marginal of $F_n$. Hence, $1-G_n$ is the empirical distribution function of $T_i$ which we will identify with the empirical distribution of $T_i$ in the notation. Denote by $(T_1, \ldots, T_n)$ the points on which $(1-G_n)$ is supported. In order to bound $N_1(\delta/4, 1-G_n, \mathscr{W}_{\epsilon}, \bar{W}_{\epsilon})$ we need to find a family of functions $\mathscr{W}^* = \{w_1, \ldots, w_K\} \subset \mathscr{W}_\epsilon$ such that
\[
\min_{k=1, \ldots, K} (1-G_n)(|w - w_k|) \le \frac{\delta}{4} (1-G_n)(\bar W_\epsilon) = \frac{\delta}{4 \epsilon^2}.
\]
It is therefore sufficient to show that
\begin{equation*}
\min_{k=1, \ldots, K} \max_{1\leq r \leq n}\big|w(T_r)-w_k(T_r)\big|\leq \frac{\delta}{4 \epsilon^2}. 
\end{equation*}
Denoting $\delta^* = \delta / (4 \epsilon^2)$, we can construct the family $\mathscr{W}^*$ in the same way as in the first part of the proof. This yields the estimate
\[
N_1(\delta/4, 1-G_n, \mathscr{W}_{\epsilon}, \bar{W}_{\epsilon}) \le (n+1)^{M_{\delta^*}}.
\]

For the estimate of $N_1(\delta/4, F_n \otimes F, \mathscr{V}^2_\epsilon, \bar{H}_{\epsilon})$, we notice that
\[
N_1(\delta/4, F_n \otimes F, \mathscr{V}^2_\epsilon, \bar{H}_{\epsilon}) \le N_1(\delta/4, 1-G, \mathscr{W}_{\epsilon}, \bar{W}_{\epsilon}).
\]
This is a nearly the same expression as above but with $G_n$ replaced by $G$. As there is no dependence on $n$ here, it is sufficient to show that the covering number on the right-hand side is finite, that is, there is a finite family of functions $\{w_1, \ldots, w_K\} \subset \mathscr{W}_\epsilon$ such that
\[
\min_{k=1, \ldots, K} (1-G)(|w - w_k|) \le \frac{\delta}{4} (1-G)(\bar W_\epsilon) = \frac{\delta}{4 \epsilon^2}.
\]
We will do it in two steps. In the first step, we will show that one can approximate $w$ with a function $\tilde w$ with error $(1-G)(|w - \tilde w|) \le \frac12 \frac{\delta}{4 \epsilon^2}$ that is piecewise constant and can be identified with a non-decreasing function on a finite (discrete) probability space $(\Omega^\mu, \mu)$. In the second step, we will use similar arguments as above to show that there is a family of functions $\tilde{\mathscr{W}}^* := \{ \tilde w_1, \ldots, \tilde w_K\}$ on $(\Omega^\mu, \mu)$, $\Omega^\mu \subset [0, \Tmax]$, which are non-decreasing and taking values in $[1,\epsilon^{-2}]$ and such that
\begin{equation}\label{eqn:tildeW}
\min_{k=1, \ldots, K} \mu(|\tilde w - \tilde w_k|) \le \frac{1}{2} \frac{\delta}{4} \mu(\bar W_\epsilon) = \frac{\delta}{8 \epsilon^2}.
\end{equation}

For the first step, let $\delta' = \delta / 8$. Denote by $a_1, \ldots, a_{M-1}$ atoms of the distribution $G$ with probability greater or equal to $\delta'$, i.e., $G(a_m-) - G(a_m) \ge \delta'$, with $a_M = \Tmax$ (irrespective whether there is an atom there). Denote by $I_l$, $l=1, \ldots, L$, intervals of the form $(b_l, c_l)$ or $[b_l, c_l)$ (depending whether $b_l$ is the atom already handled above) with the probability bounded from above by $\delta'$, i.e., $\int_{I_l} d(1-G(s)) \le \delta'$ (in other words, $G(c_l-) - G(b_l) \le \delta'$ or $G(c_l-) - G(b_l-) \le \delta'$ depending on whether $b_l \in I_l$). The intervals are chosen in such a way that they are disjoint between themselves and disjoint with the selected atoms $a_1, \ldots, a_M$ and exhaust $[0, \Tmax]$:
\[
\bigcup_{l=1}^L I_l \cup \{a_1, \ldots, a_M\} = [0, \Tmax].
\]
Let $w \in \mathscr{W}_\epsilon$ and define
\[
\tilde w(t) =
\begin{cases}
w(a_m), & t = a_m,\\
w(b_l-), & t \in I_l.
\end{cases}
\]
We have
\begin{align*}
(1-G)(|w - \tilde w|)
&=
\sum_{l=1}^L \int_{I_l} |w(s) - \tilde w(s)| d(1-G(s)) + \sum_{m=1}^M |w(a_m) - \tilde w(a_m)| (G(a_m-) - G(a_m))\\
&=
\sum_{l=1}^L \int_{I_l} |w(s) - w(b_l-)| d(1-G(s))\\
&\le
\sum_{l=1}^L \int_{I_l} \big(w(c_l-) - w(b_l-)\big) d(1-G(s))\\
&\le \sum_{l=1}^L \int_{I_l} \big(w(c_l-) - w(b_l)\big) \delta'\\
&\le \big( w(\Tmax-) - w(0-) \big) \delta' \le \frac{\delta'}{\epsilon^2},
\end{align*}
where in the second equality we used that $w = \tilde w$ on atoms (by definition), the first inequality is because $w$ is non-decreasing whereas the second inequality is by the bound on the $1-G$ measure of intervals $I_l$ (by construction).

When evaluating expectations of functions constant on intervals $I_l$, $l=1, \ldots, L$ (such as $\tilde w$) we can view it as taking expectations with respect to a discrete probability measure. Indeed, we can define a probability measure $\mu$ on a space $\Omega^\mu = \{1, \ldots, L+M\}$ such that
\begin{itemize}
\item index $i \in \Omega^\mu$ corresponds to interval $I_l$ or atom $a_m$,
\item the intervals and atoms are ordered in an increasing way, for example (recalling that $I_l$ may not include its left end $b_l$):\\
index $1$ corresponds to $I_1$ if $a_1 > s$, $b_2 > s$, for all $s \in I_1$, or\\
index $1$ corresponds to $a_1$ if $s > a_1$ for all $s \in I_1$.
\end{itemize}
We define
\[
\mu(i) = 
\begin{cases}
G(a_m-) - G(a_m), & \text{$i$ corresponds to atom $a_m$,}\\
\int_{I_l} d(1-G(s)), & \text{$i$ corresponds to interval $I_l$.}
\end{cases}
\]
The function $\tilde w$ has a counterpart on $\Omega\mu$ defined as
\[
w^\mu(i)=
\begin{cases}
w(a_m), & \text{$i$ corresponds to atom $a_m$,}\\
w(b_l-), & \text{$i$ corresponds to interval $I_l$.}
\end{cases}
\]
We emphasise that there is a one-to-one correspondence between non-decreasing functions with values in $[1, \epsilon^{-2}]$ on $\Omega^\mu$ and functions from $\mathscr{W}_\epsilon$ which are constant on intervals $I_l$, $l=1, \ldots, L$ and that their integrals with respect to $\mu$ and $(1-G)$, respectively, are identical. This completes step one.

Denote by $\tilde{\mathscr{W}}_\epsilon$ the family of non-decreasing functions on $\Omega^\mu$ with values in $[1, \epsilon^{-2}]$. In step two, we argue that there is a family $\tilde{\mathscr{W}}^* := \{ \tilde w_1, \ldots, \tilde w_K\} \subset \tilde{\mathscr{W}}_\epsilon$ such that \eqref{eqn:tildeW} holds for every $\tilde w \in \tilde{\mathscr{W}}_\epsilon$. The arguments to construct this family are similar as used for the computation of $N_1(\delta/4, 1-G_n, \mathscr{W}_{\epsilon}, \bar{W}_{\epsilon})$. The difference comes from the fact that the measure $\mu$ is not uniform on $\Omega^\mu$, so the number of approximation values in the construction from the beginning of the proof would depend on $\mu(i)$ for each $i \in \Omega^\mu$. Details are however too similar to merit the repetition of the whole argument, particularly that here we only need to show that the covering number is finite.

Using the one-to-one correspondence between functions on $\Omega^\mu$ and functions which are piecewise constant on $I_l$, $l=1, \ldots, L$, we are going to write $\tilde w$ and $\tilde w_k$ for both, depending on the context. With this convention, for any $w \in \mathscr{W}_\epsilon$ we have
\begin{align*}
\min_{k=1, \ldots, K} (1-G)(|w - \tilde w_k|)
&\le
\min_{k=1, \ldots, K} (1-G)(|w - \tilde w| + |\tilde w - \tilde w_k|)\\
&=
(1-G)(|w - \tilde w|) + \min_{k=1, \ldots, K} (1-G)(|\tilde w - \tilde w_k|)\\
&=
(1-G)(|w - \tilde w|) + \min_{k=1, \ldots, K} \mu (|\tilde w - \tilde w_k|)
\le
\frac{\delta}{8\epsilon^2} + \frac{\delta}{8\epsilon^2} = \frac{\delta}{4} \mu(\bar W_\epsilon),
\end{align*}
which completes the proof that the covering number $N_1(\delta/4, 1-G, \mathscr{W}_{\epsilon}, \bar{W}_{\epsilon})$ is finite.

    %===================
 	
Finally, we need to show that the third assumption  of Theorem~\ref{theo:conv_ustat} is satisfied. As in \eqref{eqn:sum_covering}
\begin{equation*}
N_1(\delta, F \otimes F, \bar{\mathscr{H}}_{\epsilon}, \bar{H}_{\epsilon})
\leq
N_1(\delta/4, F \otimes F, \mathscr{V}^1_{\epsilon}, \bar{H}_{\epsilon})\ N_1(\delta/4, F \otimes F, \mathscr{V}^2_\epsilon, \bar{H}_{\epsilon})
\end{equation*}
and we only need to show that the right-hand side is finite. As above, we have
\begin{align*}
N_1(\delta/4, F \otimes F, \mathscr{V}^1_{\epsilon}, \bar{H}_{\epsilon}) & \le N_1(\delta/4, 1-G, \mathscr{W}_{\epsilon}, \bar{W}_{\epsilon}),\\
N_1(\delta/4, F \otimes F, \mathscr{V}^2_{\epsilon}, \bar{H}_{\epsilon}) & \le N_1(\delta/4, 1-G, \mathscr{W}_{\epsilon}, \bar{W}_{\epsilon}).
\end{align*}
We have already shown that the covering number on the right-hand side is finite, so the proof of the third assumption of Theorem \ref{theo:conv_ustat} is complete.
    
Because all assumptions of Theorem~\ref{theo:conv_ustat} are satisfied, as $n \to \infty$,
\begin{equation}\label{resultI}
        \text{sup}_{g\in \mathcal{G}_{\epsilon}} \bigg| \frac{2}{n(n-1)} \sum_{1\leq i<j\leq n}\bar{h}^{g}(\mathcal{x}_i,\mathcal{x}_j) - F\otimes F(\bar{h}^g)\bigg| \overset{a.s.}{\to} 0.
\end{equation}
\end{proof}
%--------------------------
 
We have demonstrated the uniform convergence of the symmetrised version of the numerator of $\tdu$. Under the same assumptions and using analogous arguments, we can show the uniform convergence of the denominator of $\tdu$ because the difference is only in the indicator function $\indd{S(T_i;Z_i)<S(T_i;Z_j)}$. We state it formally in the following corollary.
 
 %------------------------
\begin{corollary}[Convergence of the Symmetrised Denominator of Time-dependent Uno's C-index]\label{corollary:conv_den_tdu}
Let $\mathcal{x}_i=(\dZ_i,T_i,D_i)$, $i=1,2,\ldots$, be independent samples from distribution $F$ on $\mathscr{X}$. Recall the terms inside the summation in the denominator of \eqref{eqn:tdu}, and rewrite the terms as follows
\begin{equation}\label{eqn:psi_g}
 		\psi^{g}(\mathcal{x}_i,\mathcal{x}_j)=\indd{T_i < D_i}\indd{T_i<X_j}g^{-2}(T_i), \qquad g \in \mathcal{G}_{\epsilon}.
\end{equation}
Assuming conditions (R1-R2), we have
\begin{equation*}
 		\displaystyle \text{sup}_{g\in \mathcal{G}_{\epsilon}} \bigg| \frac{2}{n(n-1)} \sum_{1\leq i<j\leq n}\bar{\psi}^g(\mathcal{x}_i,\mathcal{x}_j) - F\otimes F(\bar{\psi}^g)\bigg| \overset{a.s.}{\to} 0,
\end{equation*}
as $n \to \infty$ for any $\epsilon > 0$, where $\bar{\psi}^{g}$ is the symmetrised version of $\psi^g$.
\end{corollary}
 %------------------------

In Lemma~~\ref{lem1} and Corollary~\ref{corollary:conv_den_tdu}, we have demonstrated the convergence of U-statistics whose kernels are the symmetrised terms inside the summations in the numerator and denominator of $\tdu$, respectively. However, we need the convergence of the non-symmetrised terms inside the summations. Since the convergence of the symmetrised functions results in the convergence of the respective non-symmetrised functions, the lemma and corollary guarantee the convergence of the quantities of interest. The lemma and corollary results will be used to prove Theorem~\ref{theo:conv_tdu}, which is given in the following proof.

 %----------------------------------
\begin{proof}[Proof of Theorem~\ref{theo:conv_tdu}]
Recall that $\mathcal{x}_i=(\dZ_i,T_i,D_i)$ for $i=1,2,\ldots$, is a sample from the distribution $F$, and $h^g$, $g \in \mathcal{G}_\epsilon$ equals
    \begin{equation*}
        h^g(\mathcal{x}_i,\mathcal{x}_j)=\indd{T_i < D_i}\indd{S(T_i;\dZ_i) <S(T_i;\dZ_j)}\indd{T_i<X_j,T_i < \Tmax}g^{-2}(T_{i}),
    \end{equation*}

The numerator of the statistic $\text{C}^w_n$ equals
\[
\sum_{i \ne j}^n h^{\widehat G_n}(x_i, x_j)
=
\frac{2}{n(n-1)}\sum_{1\leq i<j\leq n}\bar{h}^{\widehat G_n}(\mathcal{x}_i,\mathcal{x}_j),
\]
where 
\begin{equation*}\label{eqn:hg_sim}
\bar{h}^g\big(\mathcal{x}_i,\mathcal{x}_j\big)=
	\frac{1}{2}\big(h^g(\mathcal{x}_i,\mathcal{x}_j)+h^g(\mathcal{x}_j,\mathcal{x}_i)\big)
\end{equation*}
is the symmetrised version of $h^g$. 

We have
\begin{align*}
&\bigg|\frac{2}{n(n-1)}\sum_{1\leq i<j\leq n}\bar{h}^{\widehat G_n}(\mathcal{x}_i,\mathcal{x}_j) - F\otimes F(\bar{h}^{G})\bigg|\\
&\le
\bigg|\frac{2}{n(n-1)}\sum_{1\leq i<j\leq n}\bar{h}^{\widehat G_n}(\mathcal{x}_i,\mathcal{x}_j) - F\otimes F(\bar{h}^{\widehat G_n})\bigg|
+
\bigg|F\otimes F(\bar{h}^{\widehat G_n}) - F\otimes F(\bar{h}^{G}) \bigg|\\
&\le
\sup_{g \in \mathcal{G}_\epsilon}\bigg|\frac{2}{n(n-1)}\sum_{1\leq i<j\leq n}\bar{h}^{g}(\mathcal{x}_i,\mathcal{x}_j) - F\otimes F(\bar{h}^{g})\bigg|
+
\bigg|F\otimes F(\bar{h}^{\widehat G_n}) - F\otimes F(\bar{h}^{G}) \bigg|.
\end{align*}
By Lemma~\ref{lem1},  the first term above converges a.s. to $0$ as $n \to \infty$. For the convergence of the second term, we apply the dominated convergence theorem and almost sure the uniform convergence of $\widehat G_n$ to $G$ assumed in the statement of the theorem. In conclusion, we have
\begin{equation}\label{eqn:conv1}
\frac{2}{n(n-1)}\sum_{1\leq i<j\leq n}\bar{h}^{\widehat G_n}(\mathcal{x}_i,\mathcal{x}_j)  \overset{a.s.}{\to} F\otimes F(\bar{h}^{G})
\end{equation}
as $n \to \infty$. 

We argue similarly for the denominator of $\text{C}^w_n$ which is equal to
\[
\frac{2}{n(n-1)} \sum_{1\leq i<j\leq n}\bar{\psi}^{\widehat G_n}(\mathcal{x}_i,\mathcal{x}_j),
\]
where $\bar\psi^g$ is the symmetrisation of $\psi^{g}(\mathcal{x}_i,\mathcal{x}_j)=\indd{T_i < D_i}\indd{T_i<X_j}g^{-2}(T_i)$. We have
\begin{align*}
&\bigg| \frac{2}{n(n-1)} \sum_{1\leq i<j\leq n}\bar{\psi}^{\widehat G_n}(\mathcal{x}_i,\mathcal{x}_j) - F\otimes F(\bar{\psi}^G)\bigg|\\
&\le
\bigg| \frac{2}{n(n-1)} \sum_{1\leq i<j\leq n}\bar{\psi}^{\widehat G_n}(\mathcal{x}_i,\mathcal{x}_j) - F\otimes F(\bar{\psi}^{\widehat G_n})\bigg|
+
\bigg| F\otimes F(\bar{\psi}^{\widehat G_n}) - F\otimes F(\bar{\psi}^G)\bigg|\\
&\le
\text{sup}_{g\in \mathcal{G}_{\epsilon}} \bigg| \frac{2}{n(n-1)} \sum_{1\leq i<j\leq n}\bar{\psi}^g(\mathcal{x}_i,\mathcal{x}_j) - F\otimes F(\bar{\psi}^g)\bigg|
+
\bigg| F\otimes F(\bar{\psi}^{\widehat G_n}) - F\otimes F(\bar{\psi}^G)\bigg|
\end{align*}
The first term on the right-hand side converges to $0$ a.s. as $n\to\infty$ by Corollary \ref{corollary:conv_den_tdu}. For the convergence of the second term, we use again the dominated convergence theorem.

Finally, we apply the Continuous Mapping Theorem \citep{shao2003mathematical} to conclude that 
    \begin{equation}\label{eqn:conv_tdu}
        \frac{\sum_{1\leq i<j\leq n}h^{\widehat G_n}(\mathcal{x}_i,\mathcal{x}_j)}{\sum_{1\leq i<j\leq n}\psi^{\widehat G_n}(\mathcal{x}_i,\mathcal{x}_j)} \overset{a.s.}{\to} \frac{ F\otimes F(h^G)}{\ F\otimes F(\psi^{G})}
    \end{equation}
    as $n \to \infty$. It remains to note that the right-hand side equals 
    \begin{equation}\label{eqn:rhs_conv_tdu}
	\begin{aligned}
            \frac{\ee\big[ \indd{S(T_i;Z_i)<S(T_i;Z_j)}\indd{T_i < T_j, T_i < \Tmax}\frac{\indd{T_i < D_i}}{g(T_i)}\frac{\indd{T_i < D_j}}{g(T_i)}\big]}{\ee\big[\indd{T_i < T_j, T_i < \Tmax}\indd{T_i < D_i}\frac{\indd{T_i < D_i}}{g(T_i)}\frac{\indd{T_i < D_j}}{g(T_i)}\big]}
            &=
            \frac{\prob \big[S(T_i;\dZ_i) < S(T_i;\dZ_j), T_i <T_j, T_i < \Tmax\big]}{\prob \big[T_i <T_j, T_i < \Tmax\big]}\\
            &=
            \prob\big[S(T_i;\dZ_i) < S(T_i;\dZ_j) \big|T_{i} < T_{j}, T_i < \Tmax\big].
	\end{aligned}
    \end{equation}
This completes the proof.    
\end{proof}

%-----------------------------------------
\bmsubsection{Convergence of Time-Dependent Uno's C-Index in Discrete-Time Units \label{app1.1c}}

The convergence of discrete-time models is a direct consequence of the convergence of continuous-time models. In this section, we will however rewrite results in the discrete-time setting and comment on simplifications that can be made in proofs. 

We assume that $T_i,D_i$ and $X_i$ belong to $\cT=\{1, \cdots, \Tmax\}$, the discrete set of times. The class $\mathcal{G}_\epsilon$ takes the form \eqref{eqn:G_discrete} which we recall for convenience:
\[
\widetilde{\mathcal{G}}_{\epsilon}=\{g:\{1, \cdots, \Tmax-1\}\to [\epsilon,1]:\ \text{g in non-increasing}\}.
\]

Lemma \ref{lemma:conv_num_tdu} and Corollary \ref{corollary:conv_den_tdu} hold with $\mathcal{G}_\epsilon$ replaced by $\widetilde{\mathcal{G}}_{\epsilon}$. The proof is however significantly simpler. Define
\[
\widetilde{\mathscr{W}}_\epsilon= \left\{w:\ w = g^{-2} \text { for } g \in \widetilde{\mathcal{G}}_{\epsilon}\right\}
= 
\big\{w: \{1, \cdots, \Tmax-1\}\to \left[1, 1 / \epsilon^2\right]:\ \text{$w$ non-decreasing}\big\}.
\]
All estimates of the covering numbers follow from the fact that for each $\delta^* > 0$ there is a family of functions $\{w_1, \ldots, w_K\} \subset \widetilde{\mathscr{W}}_\epsilon$ such that
\[
 \min_{k=1, \ldots, K} \max_{1\leq t < \Tmax}\big|w(t)-w_k(t)\big|\leq \delta
\]
and 
\[
K \le \Tmax^{1/(\delta\epsilon^2)}.
\]
The proof of this fact follows the construction of the family $\mathscr{W}^*$ in the first part of the proof of Lemma \ref{lemma:conv_num_tdu}. Theorem \ref{theo:conv_tdu} takes a simpler form. The proof remains identical.

\begin{lemma}[Convergence of Discrete Time-dependent Uno's C-index]\label{lem:conv_disc_tdu}
Assume the regularity conditions (R1-R2) and that there exists $\epsilon > 0$ such that $\widehat{G}_n, G \subset \widetilde{\mathcal{G}}_{\epsilon}$. If $\widehat{G}_n(t) \to G(t)$ a.s.\ as $n \to \infty$ for every $t \in \{1, \ldots, \Tmax-1\}$, then $\tdu \overset{a.s.}{\to}  \text{C}$.
\end{lemma}

%----------------------------------------

\newpage

%-----------------------------------------
\bmsection{The Model Architectures}\label{app2}
%\vspace*{12pt}
%-----------------------------------------
\begin{table*}[htp]%
\centering
\caption{Nnet-survival architectures}%
\begin{tabular*}{\textwidth}{@{\extracolsep\fill}llllll@{\extracolsep\fill}}%
\toprule
\textbf{Hyper-Parameters} & \textbf{Sim.1} & \textbf{Sim.2} & \textbf{Sim.3} & \textbf{HF Data} & \textbf{TCGA Data}\\
\midrule
\#Hidden layer & 2 & 2 & 2 & 2 & 2\\
\#Nodes in hidden layers  & (16,16) & (16,16) & (512,512) & (8,8) & (64,64)\\
\#Nodes in output layer  & 11 & 15 & 16 & 15 & 11 \\
Hidden activation function & ReLu & ReLu & ReLu & ReLu & ReLu\\
Output activation function & Sigmoid & Sigmoid & Sigmoid & Sigmoid & Sigmoid \\
$L^2$-regularisation & 0.075 & 1E-6  & 1E-6 & 1E-1 & 2E-2\\
Batch size & 250 & 50 & 50 & 52 & 250 \\
\#Epoch & 1000 & 150 & 150 & 300 & 1000 \\
Learning rate & 1E-3 & 1E-2 & 1E-3 & 1E-2 & 1E-3\\
Early Stopping (patience) & Yes (20) & Yes (20) & Yes (20) & Yes (20) & Yes (20)\\
Lower Bound ($\epsilon$) of $\widehat{G}_n$ & 0.02 & 0.02 & 0.02 & 0.02   & 0.02\\
\bottomrule
\end{tabular*}
\label{table:tab1_app}
\end{table*}

%-----------------------------------------
\begin{table*}[ht]%
\centering
\caption{Probability distribution of censoring at each period}%

\begin{tabular*}{\textwidth}{@{\extracolsep}ll@{\extracolsep\fill}}%
\toprule
\textbf{Censoring Rate} & \boldmath$(p_1,p_2,\cdots,p_{15})$ \\
\midrule
0\%    & (0, 0, 0, 0, 0 , 0, 0, 0, 0, 0, 0, 0, 0, 0, 1) \\
9\%    & (0.001, 0.009,	0.01, 0.01, 0.05, 0.03,	0.03, 0.01, 0.1, 0.1, 0.1, 0.15, 0.2, 0.2, 0) \\
40\%   & (0.001, 0.2, 0.15,	0.1, 0.1, 0.009, 0.01, 0.01, 0.05, 0.03, 0.03, 0.01, 0.1, 0.2, 0)\\
56\%   & (0.25 , 0.15, 0.15, 0.1, 0.05, 0.03, 0.03, 0.01, 0.01,	0.01, 0.001, 0.009, 0.1, 0.1, 0)\\
64\%   & (0.3, 0.2,	0.15, 0.1, 0.05, 0.03, 0.03, 0.01, 0.01, 0.01, 0.001, 0.009, 0.05, 0.05, 0)\\
73\%   & (0.4, 0.2,	0.15, 0.1, 0.05, 0.03, 0.01, 0.01, 0.01, 0.01, 0.001, 0.009, 0.01, 0.01,	0)\\
\bottomrule
\end{tabular*}
\label{table:tab2_app}
\end{table*}
%-----------------------------------------

\begin{table*}[h]%
\centering
\caption{DTSF architectures}%
\begin{tabular*}{\textwidth}{@{\extracolsep\fill}lll@{\extracolsep\fill}}%
\toprule
\textbf{Hyper-Parameters} & \textbf{HF Data} & \textbf{TCGA Data} \\
\midrule
Minimum node size & (3, 10, 25, 50, 75, 100, 125, and 150) & (3, 5, 10, 25, 5, 150, 250, 500, and 1000) \\
\#Trees  & 100 & 100 \\
\#Variables to possibly split at in each node & Rounded
		down square root of the number variables & Rounded
		down square root of the number variables \\
Splitting rule & Hellinger & Hellinger\\
Maximal tree depth   & Unlimited & Unlimited\\
Grow a probability forest &  Yes  &  Yes \\
Lower Bound ($\epsilon$) of $\widehat{G}_n$ & 0.02 & 0.02\\
\bottomrule
\end{tabular*}
\label{table:tab3_app}
\end{table*}
%-----------------------------------------

\newpage
%----------------------------------------
\bmsection{Implementation Results}\label{app3}
%----------------------------------------
\begin{figure*}[h]
\centering
\includegraphics[width=400pt,height=14pc]{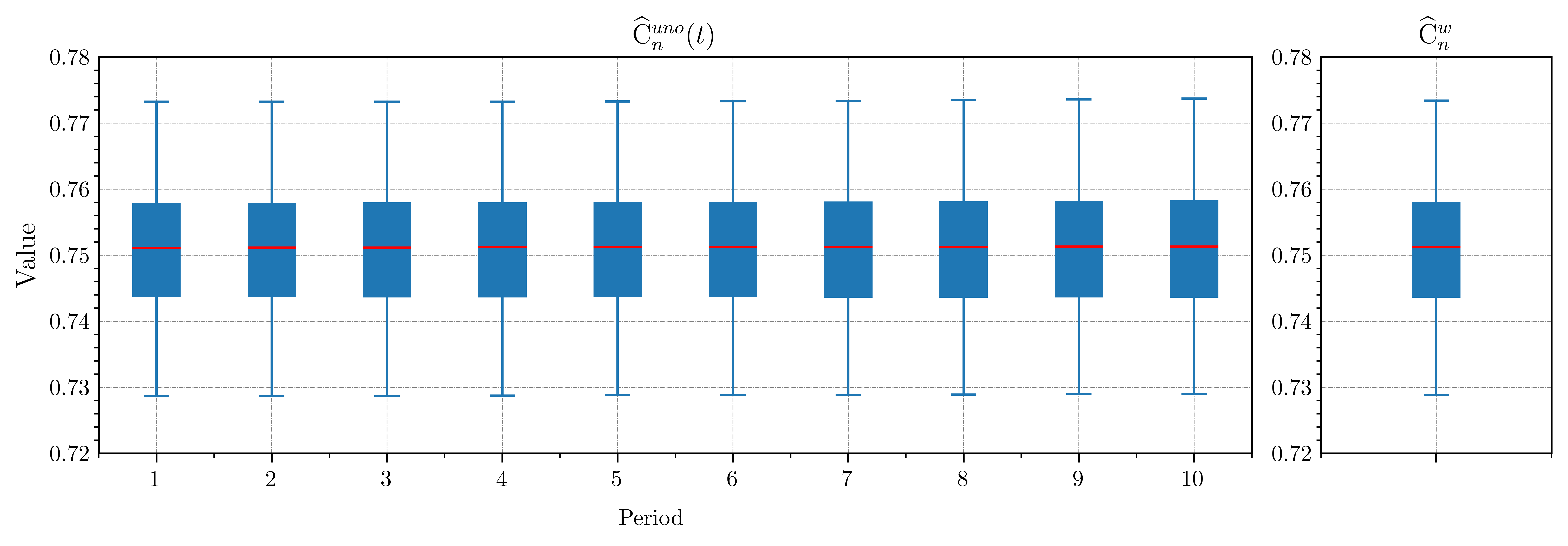}
\caption{ Uno's C-index (Left) and time-dependent Uno's C-index (Right) from a Nnet survival architecture (`Sim.1' column of Table~\ref{table:tab1_app} in Appendix~\ref{app2}) fitted to the PH data over \{1, $\cdots$, $\Tmax -1$\}. They were computed on 100 independent fully uncensored test data ($n_\text{test}$=1000) from a fixed fully uncensored train data ($n_\text{train}$=1000).}
\label{figc1}
\end{figure*}

\end{document}